\documentclass[final]{l4dc2026}
\usepackage{comment}
\usepackage{algpseudocode}

\usepackage[most]{tcolorbox}

\definecolor{paperred}{RGB}{140,30,30}
\definecolor{paperredbg}{RGB}{250,245,245}

\newcounter{algorithm}
\renewcommand{\thealgorithm}{\arabic{algorithm}}

\definecolor{alggray}{RGB}{240,240,240}

\newtcolorbox{algobox}[2][]{
colback=alggray,
colframe=black!60,
fonttitle=\bfseries,
title={Algorithm~\thealgorithm: #2},
boxrule=0.6pt,
arc=2pt,
#1
}

\usepackage{ifthen}
\newboolean{showcomments}
\setboolean{showcomments}{true}
\usepackage{todonotes}

\definecolor{bleudefrance}{rgb}{0.19, 0.55, 0.91}
\definecolor{ao(english)}{rgb}{0.0, 0.9, 0.0}

\newcommand{\addcite}[0]{\ifthenelse{\boolean{showcomments}}
{\textcolor{purple}{(add cite(s)) }}{}}%

\newcommand{\emmargin}[1]{\ifthenelse{\boolean{showcomments}}{\marginpar{\color{bleudefrance}\tiny EM: #1}}{}}

\newcommand{\enrique}[1]{  \ifthenelse{\boolean{showcomments}}
{\todo[inline,caption={},color=bleudefrance]{Enrique: #1}}{}}
\newcommand{\mahyar}[1]{  \ifthenelse{\boolean{showcomments}}
{\todo[inline,caption={},color=orange]{Mahyar: #1}}{}}
\newcommand{\jared}[1]{  \ifthenelse{\boolean{showcomments}}
{\todo[inline,caption={},color=green]{Jared: #1}}{}}

\newcommand{\agustin}[1]{  \ifthenelse{\boolean{showcomments}}
{\todo[inline,caption={},color=purple!60!white]{Agustin: #1}}{}}

\newboolean{showedits}
\setboolean{showedits}{false}
\usepackage[markup=underlined,commandnameprefix=ifneeded]{changes}
\definechangesauthor[color=bleudefrance]{EM}
\newcommand{\aem}[1]{
\ifthenelse{\boolean{showedits}}
{\added[id=EM]{#1}}
{\!#1\hspace{-4.75pt}}
}
\newcommand{\repem}[2]{
\ifthenelse{\boolean{showedits}}
{\replaced[id=EM]{#1}{#2}}
{\!#1\hspace{-4.75pt}}
}
\newcommand{\dem}[1]{
\ifthenelse{\boolean{showedits}}
{\deleted[id=EM]{#1}}
{}
}

\newboolean{revisionmode}
\setboolean{revisionmode}{false}
\newcommand{\revise}[1]{
\ifthenelse{\boolean{revisionmode}}
{{\color{red!70!black}#1}}
{#1}
}

\usepackage{bibentry}

\usepackage{enumitem}

\usepackage{subcaption}

\DeclareMathOperator*{\argmin}{arg\,min}

\newcommand{\bx}[0]{\mathbf{x}}
\newcommand{\bu}[0]{\mathbf{u}}
\newcommand{\by}[0]{\mathbf{y}}

\newcommand{\bj}[0]{\mathbf{J}}
\newcommand{\jub}[0]{J_{\operatorname{ub}}^\lambda}
\newcommand{\jlb}[0]{J_{\operatorname{lb}}^\lambda}

\makeatletter
\def\tagform@#1{\maketag@@@{\color{blue}(#1)}} 
\makeatother

\newtheorem{theorem}{Theorem}
\newtheorem{lemma}{Lemma}
\newtheorem{proposition}{Proposition}
\newtheorem{corollary}{Corollary}
\newtheorem{definition}{Definition}
\newtheorem{remark}{Remark}
\newtheorem{assumption}{Assumption}

\title[Data-driven acceleration of MPC]{Data-driven Acceleration of MPC with Guarantees}
\usepackage{times}

\author{%
 \Name{Agustin Castellano} \Email{acaste11@jhu.edu}
 \AND
 \Name{Shijie Pan} \Email{span34@jhu.edu}
 \AND
 \Name{Enrique Mallada} \Email{mallada@jhu.edu}~~\\
 \addr Dept. of Electrical and Computer Engineering\\
 Johns Hopkins University%
}

\begin{document}
\setlist[enumerate]{label=\textcolor{blue}{\arabic*.}}

\maketitle

\begin{abstract}%
 Model Predictive Control (MPC) is a powerful framework for optimal control but can be too slow for low-latency applications. We present a data-driven framework to accelerate MPC by replacing online optimization with a nonparametric policy constructed from offline MPC solutions. Our policy is greedy with respect to a constructed upper bound on the optimal cost-to-go, and can be implemented as a nonparametric lookup rule that is orders of magnitude faster than solving MPC online. Our analysis shows that under sufficient coverage conditions of the offline data, the policy is recursively feasible and admits provable, bounded optimality gap. These conditions establish an explicit trade-off between the amount of data collected and the tightness of the bounds.\revise{~New solutions can be incorporated straightforwardly without the need for retraining, enabling continual improvement.}Our experiments show that this policy is between $100$ and $1000$ times faster than standard MPC, with only a modest hit to optimality, showing  potential for real-time control tasks.
\end{abstract}


\section{Introduction}



Model Predictive Control (MPC) is a powerful framework for optimal control of constrained, high-dimensional dynamical systems \citep{mayne2014model}. Born from the process control industry in the late 1970s \citep{richalet1978model, cutler1980dynamic}, it has since matured and been applied to myriad of different industries and applications, including aerospace \citep{di2018real}, automotive \citep{hrovat2012development}, thermal control in buildings \citep{drgovna2020all}, and more \citep{forbes2015model}.
At the core of MPC is the solution approach of iteratively solving a receding-horizon constrained optimization problem \citep{mayne2000constrained}. Trade-offs between the problem horizon and the quality of the solutions have been extensively studied \citep{grune2010analysis,reble2012unconstrained,worthmann2012estimates}. Notwithstanding, one of the core challenges of MPC is how to get \emph{fast, high-quality} controls in an online setting.

Ways of speeding up MPC have been explored for decades \citep{garcia1989model}. In the case of linear systems with quadratic costs and polytopic constraints, it is a well-known fact that the optimal controller is piecewise affine \citep{bemporad2002explicit}. A substantial body of work leverages this fact in \emph{explicit} MPC \citep{alessio2009survey}, where one seeks to learn the optimal controller for each polyhedral region, by combining neural networks with projection schemes \citep{chen2018approximating} or with multiparametric quadratic programming \citep{maddalena2020neural}.  
\revise{These projection schemes, however, are costly and require solving additional optimization problems. Other works enable explicit (nonlinear) MPC via neural networks \citep{hertneck2018learning}, or by nonparametric methods, for example by using set membership approximation \citep{canale2009set}, kernel regression \citep{carnerero2023kernel, huang2023robust}, quasi-interpolation \citep{ganguly2025explicit}, nonlinear piecewise approximations \citep{trinh2016explicit} and {tube} MPC \citep{bayer2016tube}. 
One drawback of neural network methods \citep{chen2018approximating, hertneck2018learning} is that they are difficult to adapt to new data, often requiring retraining from scratch. Some of the nonparametric methods cited \citep{canale2009set} only guarantee stability/performance \emph{asymptotically}, while others \citep{carnerero2023kernel, trinh2016explicit} provide no guarantees or rely on too strong assumptions on the system dynamics \citep{bayer2016tube}.}


\revise{In contrast with the afforementioned literature, we propose a novel framework that uses offline solutions to derive a nonparametric policy with \textbf{non-asymptotic performance and feasibility guarantees}. New data modifies our policy locally, enabling \textbf{continual improvement}: new data is incorporated with no retraining needed and without appreciable performance degradation.}
Specifically, we make the following contributions:
\vspace{-6pt}
\begin{enumerate}[itemsep=0pt]
    \item We present a novel nonparametric policy that approximately solves MPC problems. This policy is built with offline data from a more constrained MPC solution.
    \item \revise{This policy can learn continually, with new solutions added without any retraining.}
    \item We establish conditions on the offline data that ensure recursive feasibility and bounds on the optimality gap over the whole domain.
    \item Empirically, we show that this policy can be implemented efficiently \revise{as a lookup rule on a CPU}, and during inference is orders of magnitude faster than online MPC. 
\end{enumerate}
The rest of the paper is organized as follows. \textbf{Section} \ref{sec:preliminaries} reviews standard MPC. \textbf{Section} \ref{sec:conservative-problem} presents the conservative problem we solve. \textbf{Section} \ref{sec:nonparametric-policy} defines the data-driven nonparametric policy, and establishes sufficient data-coverage conditions that guarantee feasibility and a desired performance. We present two algorithms in \textbf{Section} \ref{sec:algorithms}: one based on stochastic sampling (Algorithm \ref{alg:1}) and another one based on sequential splitting of the state-space domain (Algorithm \ref{alg:verification}), that upon termination provide guarantees of recursive feasibility and suboptimality.
Experiments in \textbf{Section} \ref{sec:experiments} show our verification algorithm in action and empirically contrast trade-offs between performance and controller latency for our method versus standard MPC.


%
\section{Preliminaries}\label{sec:preliminaries}
We are interested in solving the following problem:
\begin{subequations}
\label{eq:og-prob}
\begin{align}
    J(\bx_0) \triangleq  \min_{\bu_{0:T-1}}&\sum_{t=0}^{T-1}\gamma^t c(\bx_t, \bu_t) + F(\bx_T) \label{eq:og-prob-a}\\ 
    \text{subject to:}~~& \bx_{t+1}=f(\bx_t, \bu_t),& t=0,\ldots,T-1\label{eq:og-prob-b}\\
    &\bx_t\in\mathbb{X}, & t=1,\ldots,T-1\label{eq:og-prob-c}\\
    & \bu_t\in\mathbb{U}, & t=0,\ldots,T-1\label{eq:og-prob-d} 
\end{align}
\end{subequations}
where $\bu_{0:T-1}\triangleq\left[\bu_0~\bu_1\ldots \bu_{T-1}\right]$ is the sequence of controls, the problem horizon satisfies $1\leq T\leq \infty$, and $\gamma\in(0, 1]$ is a discount factor. Stage costs $c(\cdot,\cdot)$ are nonnegative, the feasible sets are {compact} and satisfy $\mathbb{X}\subseteq\mathbb{R}^n$, $\mathbb{U}\subseteq\mathbb{R}^m$. Terminal state constraints (if any) are encoded via $F:\mathbb{X}\to{\mathbb{R}_{\geq 0}}\cup\{+\infty\}$.
Without loss of generality, we assume that the origin is an equilibrium point of $f$, i.e. $\mathbf{0} = f(\mathbf{0},\mathbf{0}) \in \mathbb{X}\;$. Further, we make the following additional assumptions. 
\begin{assumption}[Lipschitz dynamics] \label{assn:lipschitz-dynamics}
The map $f$ is $L_f-$Lipschitz in $\bx$ and $L_u-$Lipschitz in $\bu$:
    \begin{equation}
       \revise{\left\| f(\bx, \bu) - f(\bx', \bu') \right\|_{\mathbb{X}}} \leq L_f  \left\|\bx - \bx'\right\|_{\mathbb{X}} +  L_u\left\|\bu - \bu'\right\|_{\mathbb{U}}\quad\forall \bx, \bx' \in \mathbb{X},~~\forall \bu, \bu' \in \mathbb{U}\;,
    \end{equation}
    where $\|\cdot\|_{\mathbb{X}}$ and $\|\cdot\|_{\mathbb{U}}$ are appropriate norms on $\mathbb{X}$ and $\mathbb{U}$.
\end{assumption}
We omit the dependence of $\|\cdot\|_\square$ on $\mathbb{X}$ and $\mathbb{U}$ when it is clear from context.
\begin{assumption}[Stationarity of optimal solutions]\label{assn:stationarity}
    The optimal policy $\pi^\star:\mathbb{X}\to\mathbb{U}$ for \eqref{eq:og-prob} is stationary.
\end{assumption}%
It follows from Assumption \ref{assn:stationarity} that the optimal cost-to-go $J(\cdot)$ is also stationary, and satisfies the Bellman Equation \citep{bellman1954theory, bertsekas2012dynamic}:
\begin{equation}
    J(\bx) = c\left(\bx, \pi^\star(\bx)\right) + \gamma \cdot J\left(f(\bx,\pi^\star(\bx)\right). \label{eq:bellman}
\end{equation}%
The preceding assumption is not overly restrictive: it captures (among other cases) infinite-horizon LQR control \citep[Ch. 3]{bertsekas2012dynamic} and shortest path problems \citep[Ch. 2]{bertsekas2012dynamic}.

\subsection*{Typical MPC solution scheme {and limitations}}
In Model Predictive Control, the standard approach is to solve \eqref{eq:og-prob} over a smaller horizon $H \ll T$, obtain the optimal sequence of controls $\bu_0,\ldots,\bu_{H-1}$, apply only the first one ($\bu_0$) to the system. Then, this procedure is repeated from the new state. This approach is also called receding horizon control \citep{mattingley2011receding, rawlings2002stability}. It provides a trade-off between computation time (smaller $H$) at the expense of solution quality (larger $H$). It always produces a stationary controller \citep[Section 3.8.1.]{mayne2000constrained}, although one of its main drawbacks is that, by its own, it does not guarantee recursive feasibility \citep{lofberg2012oops}. It is common in the MPC literature to add a terminal constraint $\bx_H\in \mathbb{X}_H$, with $\mathbb{X}_H$ being a constraint-satisfying control invariant set \citep{chen1998quasi,mayne2000constrained}, or to carefully design a terminal cost (like $F(\cdot)$ in \eqref{eq:og-prob-a}) that yields recursive feasibility (e.g. using a Control Lyapunov function as $F(\cdot)$ \citep{jadbabaie2002unconstrained}). However, computing control invariant sets or Lyapunov functions for general non-linear systems is a hard problem in and of itself \citep{blanchini1999set}.

In this work, we overcome these limitations by using offline, precomputed solutions to a slightly conservative version of \eqref{eq:og-prob}. These solutions are used to define our data-driven policy, which will enjoy---by design and with sufficient data---recursive feasibility. 

\begin{figure}[t]
    \centering
    \includegraphics[width=\linewidth]{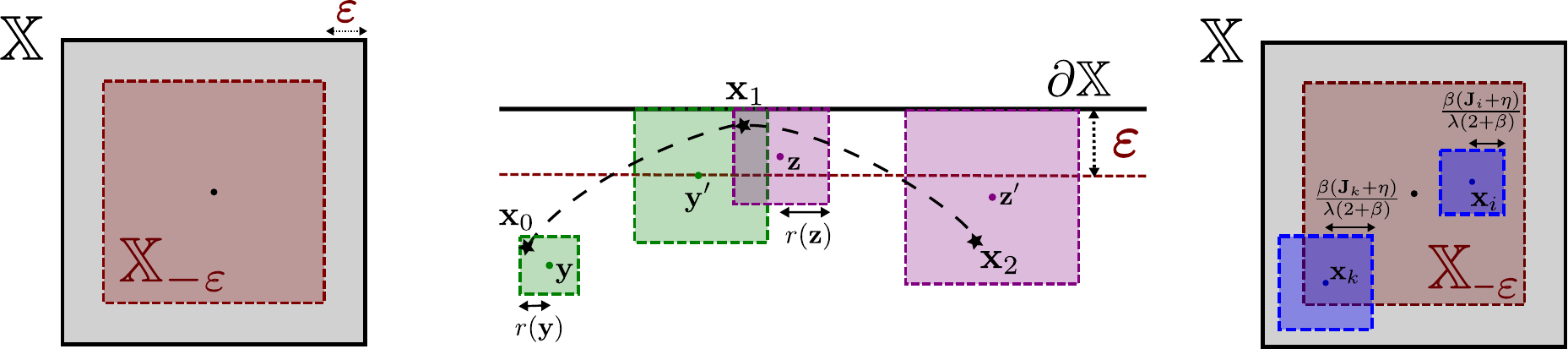}
    \caption{Left: Original constraint set $\mathbb{X}$ and its erosion $\mathbb{X}_{-\varepsilon}$. Middle: Feasibility certificates under our framework. The trajectory $\left(\bx_0, \bx_1, \bx_2, \ldots\right)$ marked with `$\star$' is produced by our policy. Optimal transitions $\mathbf{y}\overset{\pi^\star}{\to}\mathbf{y}'$ and $\mathbf{z}\overset{\pi^\star}{\to}\mathbf{z}'$ are precomputed offline (by solving \eqref{eq:conservative-prob}) and stored in a dataset $\mathcal{D}$. The control associated with each state (e.g. $\mathbf{y}$) in the dataset is also feasible in a neighborhood of that point (the ball with radius $r(\mathbf{y})$, see Prop. \ref{prop:local-feasibility}). Right: Performance guarantees for our policy (Theorem \ref{thm:performance-guarantees}). Each triplet $(\bx_i,\bu_i,\bj_i)$ in the dataset certifies a ball $\mathbb{B}\left(\bx_i, \tfrac{\beta\left(\bj_i+\eta\right)}{\lambda\left(2+\beta\right)}\right)$, wherein $\tfrac{J^\pi(\bx)-J(\bx,\varepsilon)}{J(\bx,\varepsilon)+\eta} \leq \beta$ for any $\bx$ in that ball.} 
    \label{fig:2in1}
\end{figure}

\section{Offline solution strategy} \label{sec:conservative-problem}
During the offline phase, we collect data by solving a more conservative version of Problem \eqref{eq:og-prob}, which we describe now. %
Let $\mathbb{B}_{\varepsilon}$ be the $\varepsilon$-ball in $\mathbb{R}^n$ centered at the origin. We define the \emph{erosion} of $\mathbb{X}$ at level $\varepsilon$ as:
\begin{equation*}
\mathbb{X}_{-\varepsilon}\triangleq \mathbb{X}\ominus\mathbb{B}_{\varepsilon} = \left\{\bx\in\mathbb{R}^n: \bx + \mathbb{B}_{\varepsilon}\subseteq\mathbb{X}\right\}\;,
\end{equation*}
where $\ominus$ denotes Pontryagin/Minkowski difference \citep[Ch. 3]{blanchini2008set} (see Figure \ref{fig:2in1}). Consider the following generalization of problem \eqref{eq:og-prob}:
\begin{subequations}
\label{eq:conservative-prob}
\begin{align}
    J(\bx_0, \varepsilon) \triangleq  \min_{\bu_{0:T-1}}&\eqref{eq:og-prob-a}\\ 
    \text{subject to:}~~& \eqref{eq:og-prob-b}, \eqref{eq:og-prob-d}, {\bx_t\in\mathbb{X}_{-\varepsilon}}, & t=1,\ldots,T-1  \label{eq:conservative-prob-xt}
\end{align}
\end{subequations}
The cost-to-go $J(\cdot,\cdot)$ above is a mapping from $\mathbb{X}\times\mathbb{R}_{\geq 0}$ to $\mathbb{R}_{\geq 0}\cup\{+\infty\}$. Note that $J(\cdot, 0)$ reduces to $\eqref{eq:og-prob}$. We make two additional assumptions.
\begin{assumption}[Shrunk problem is feasible]\label{assn:shrunk-feasible}
    There exists a small, positive $\varepsilon$ such that Problem \eqref{eq:conservative-prob} is feasible for all $\bx_0\in\mathbb{X}$.
\end{assumption}
We highlight that the assumption above is for any $\bx_0\in\mathbb{X}$, and not for any $\bx_0\in\mathbb{X}_{-\varepsilon}$. This means initial conditions $\bx_0\in\mathbb{X}\setminus\mathbb{X}_{-\varepsilon}$ are required to ``jump in'' to $\mathbb{X}_{-\varepsilon}$ in one step. The trajectory $\mathbf{z}\to\mathbf{z}'$ in Figure \ref{fig:2in1} is an example of this behavior.\footnote{\revise{This assumption can be relaxed by enforcing that the system enters $\mathbb{X}_{-\varepsilon}$ after at most $K$ steps, i.e. changing \eqref{eq:conservative-prob-xt} to $\bx_t\in\mathbb{X}_{-\varepsilon}~\forall t\geq K$, for some $K:1\leq K <T$. Our results can be easily generalized to the case $K>1$, which is related to notions of \emph{control recurrence} \citep{shen2022model, sibai2026recurrence}. For ease of exposition we focus on the case $K=1$ and defer the general one for future work.}}
\begin{assumption}[Cost-to-go of conservative problem is  Lipschitz]\label{assn:J-locally-lipschitz}
There exists $L>0$ and $L_J>0$ such that for any $\varepsilon$ satisfying the assumption above, we have:
\vspace{-5pt}
    \begin{enumerate}
    [itemsep=0pt]
    \item\label{assn:J-lip-1} 
    $
    J(\bx_0, \varepsilon)-J(\bx_0, 0) \leq L\cdot \varepsilon\;\quad \forall~\bx_0\in\mathbb{X}\;.
    $
    \item \label{assn:J-lip-2} 
    $
    J(\bx_0,\varepsilon)-J(\bx_0',\varepsilon) \leq L_J \|\bx_0-\bx_0'\|\quad \forall \bx_0,\bx_0'\in\mathbb{X}\;.
    $
    \end{enumerate}
\end{assumption}
Assumption \ref{assn:J-locally-lipschitz}.\ref{assn:J-lip-1} and connections as to whether the map $\left(\bx_0,\varepsilon\right)\mapsto \bu_{0:T-1}^\star$ is locally Lipschitz are concepts related to perturbation analysis of optimization problems \citep{rockafellar1998variational} and notions of \emph{strong stability} of solutions, see e.g. Section 5 in \cite{bonnans1998optimization}.
The middle panel in Figure \ref{fig:2in1} shows trajectories under our policy: if the dataset $\mathcal{D}$ has optimal transitions coming from \eqref{eq:conservative-prob} (i.e. satisfying $\bx_t\in\mathbb{X}_{-\varepsilon},\forall t\geq1$), then our policy is guaranteed feasible for \eqref{eq:og-prob} (i.e. $\bx_t\in\mathbb{X},~\forall t\geq 1$). Assumption \ref{assn:J-locally-lipschitz}.\ref{assn:J-lip-2} asks for the optimal cost-to-go to be Lipschitz continuous in the state variable. This  is not overly restrictive, as shown below:
\begin{proposition}[Sufficient conditions for Assumption \ref{assn:J-locally-lipschitz}.\ref{assn:J-lip-2} Lemma 3 in \cite{bucsoniu2018continuous}]
    \\Suppose Assumption \ref{assn:lipschitz-dynamics} holds, the stage cost $c(\cdot,\cdot)$ is $L_c$-Lipschitz and $\gamma \max\{L_f,L_u\} < 1$. Then Assumption \ref{assn:J-locally-lipschitz}.\ref{assn:J-lip-2} holds with
    $
    L_J \leq \frac{L_c}{1-\gamma \max\{L_f,L_u\}}\;.
    $
\end{proposition}
\section{Nonparametric policy}\label{sec:nonparametric-policy}
In this work we propose a nonparametric policy based on \emph{offline} (i.e. precomputed) solutions to \eqref{eq:conservative-prob}. Optimal tuples $\left(\bx_i, \bu_i, \bj_i\right)$ from the \emph{conservative} problem \eqref{eq:conservative-prob} are stored in a dataset $\mathcal{D}$:
\begin{equation}
    \mathcal{D} \triangleq \left\{\left(\bx_i, \bu_i, \bj_i\right)\right\}_{i} \quad\text{where~} \bu_i=\pi^\star(\bx_i), ~\bj_i = J(\bx_i,\varepsilon)\;. \label{eq:dataset}
\end{equation}    
We implement a policy that takes a ``close'' action in the dataset, as defined next.
\begin{definition}[Nonparametric policy] \label{def:npp} Given a dataset $\mathcal{D}$ as in \eqref{eq:dataset} and a parameter $\lambda > 0$, define:
    $$
    \pi_{\mathcal{D}}:\mathbb{X}\to\mathbb{U}\quad \pi_{\mathcal{D}}(\bx) = \bu_\iota, \text{~where}~~ \iota = \argmin_{1\leq i \leq |\mathcal{D}|}\big\{\bj_i + \lambda \cdot \|\bx-\bx_i\|\big\}\;.
    $$
\end{definition}
A precise value for $\lambda$ will be given later. This policy can be thought of as the one-nearest-neighbor regressor \citep[Ch. 13]{hastie2009elements} based on the data $\left\{(\bx_i,\bu_i)\right\}_i$ from $\mathcal{D}$, with a regularization factor $\frac{1}{\lambda}\bj_i$. This policy will \emph{accelerate} MPC because inference will be done at a much lower latency (details deferred until Section \ref{sec:experiments}), with only a modest hit on performance. 
\begin{remark}[Policy is built from the conservative problem]
    We highlight that the dataset $\mathcal{D}$ defined above, and hence the policy, come from solving offline the \emph{conservative} problem \eqref{eq:conservative-prob}, and not the original one \eqref{eq:og-prob}. Requiring the states $\bx_t$ to be in $\mathbb{X}_{-\varepsilon}$ for all $t>0$ will allow us to guarantee recursive feasibility of the nonparametric policy. Even though we solve a more conservative problem, we do so for the full horizon $T$ (instead of using the lookahead horizon $H$).
\end{remark}
\revise{We study conditions on $\pi_{\mathcal{D}}$ for \emph{(i)} recursive feasibility and \emph{(ii)} bounded suboptimality next.}
\subsection{Guaranteeing recursive feasibility}
The main idea to establish recursive feasibility of our policy is by leveraging the fact that we are solving the \emph{conservative} problem over \eqref{eq:conservative-prob} $\mathbb{X}_{-\varepsilon}$, meaning optimal trajectories are separated from the boundary $\partial\mathbb{X}$. We establish this condition after the following definition. 


\begin{definition}[One-step feasibility]
\label{dfn:7}
    $(\bx, \bu)$ is {one-step feasible} with respect to \eqref{eq:conservative-prob} (respectively, to \eqref{eq:og-prob}) if $\bx\in\mathbb{X}, \bu\in\mathbb{U}\text{~and~}f(\bx,\bu)\in\mathbb{X}_{-\varepsilon}$~(resp. $f(\bx,\bu)\in\mathbb{X}$).
\end{definition}

\begin{proposition}[Local feasibility] \label{prop:local-feasibility} If $(\bx,\bu)$ is one-step feasible w.r.t. \eqref{eq:conservative-prob} and $\bx'=f(\bx,\bu)$, then $(\bx_0,\bu)$ is one-step feasible w.r.t. \eqref{eq:og-prob} for all $\bx_0 \in \mathbb{B}\big(\bx, r(\bx)\big)\cap\mathbb{X}$, where:
    \begin{equation}
                r(\bx) \triangleq  \frac{\operatorname{dist}\left(f(\bx, \bu),\partial\mathbb{X}\right)}{L_f} \geq \frac{\varepsilon}{L_f} \label{eq:feasibility-radius}
    \end{equation}
\end{proposition}
The proof of the proposition above is in Appendix \ref{proof:one-step-feasibility}, and relies on the Lipschitzness of $f(\cdot,\cdot)$ (Assumption \ref{assn:lipschitz-dynamics}). The inequality in \eqref{eq:feasibility-radius} comes from the assumption that the conservative problem \eqref{eq:conservative-prob} is feasible over $\bx\in\mathbb{X}$. This proposition can be visualized in the middle panel of Figure \ref{fig:2in1}: trajectories under the same $\bu$ cannot go too far apart in one step. Conditions for recursive feasibility over the whole domain follow naturally.
\begin{proposition}[Recursive feasibility for our policy]\label{prop:recursive-feasibility} If any of the following conditions hold, then $\pi_{\mathcal{D}}$ is recursively feasible for \eqref{eq:og-prob} over $\mathbb{X}$:

\begin{enumerate}
    \item \label{prop:recursive-feasibility-1} The states $\left\{\bx_i\right\}_{i}$ in $\mathcal{D}$ forms an $\tfrac{\varepsilon}{L_f}$-cover of $\mathbb{X}$\footnote{A collection of points $\{w_i\}_{i=1,2,...,|\mathcal{D}|}$ forms an $r$-cover of a normed space $\left(\mathbb{W}, \|\cdot\|\right)$ if $\mathbb{W}\subseteq\bigcup_i \mathbb{B}(w_i, r)$.}.
    \item \label{prop:recursive-feasibility-2} $\bigcup_{i=1}^{|\mathcal{D}|}\mathbb{B}\left(\bx_i, r(\bx_i)\right) \supseteq \mathbb{X}$, where $r(\bx_i) = \frac{\operatorname{dist}\left(f(\bx_i, \bu_i),\partial\mathbb{X}\right)}{L_f}$
\end{enumerate}

\end{proposition}
Note that the former condition may be overly conservative, since it considers the smallest possible feasibility radius for any point in the dataset (like transition $\mathbf{y}\to\mathbf{y}'$ in Figure \ref{fig:2in1}). Since our policy always picks actions in the dataset $\mathcal{D}$, the constraint $\bu_t\in\mathbb{U}$ is guaranteed by design.

\subsection{Bounded suboptimality}

We now turn to bounding the optimality gap of policy $\pi_{\mathcal{D}}$. The key idea developed in this section is that, with a proper choice of regularization $\lambda$ (see Definition \ref{def:npp}), our policy's cost-to-go can be lower bounded. First, recall that under Assumption \ref{assn:J-locally-lipschitz}.\ref{assn:J-lip-2} the cost-to-go for the perturbed problem is $L_J$-Lipschitz:
$
J(\bx_0, \varepsilon) \leq J(\bx_0',\varepsilon)  + L_J \|\bx_0-\bx_0'\|\;.
$
We use this to establish a global upper bound for $J(\cdot,\varepsilon)$, based on the data in $\mathcal{D}$.
\begin{definition}[Nonparametric upper \& lower bounds on J]\label{def:J-ub}
    Let $\mathcal{D}$ be the dataset in \eqref{eq:dataset}. For any $\lambda>0$ define $\jub:\mathbb{X}\to\mathbb{R},~~\jlb:\mathbb{X}\to\mathbb{R}$:
    \begin{align*}
    \jub(\bx) \triangleq \min_{1\leq i \leq |\mathcal{D}|}\big\{\bj_i + \lambda \|\bx-\bx_i\|\big\} &
    &\quad \jlb(\bx) \triangleq \max_{1\leq i \leq |\mathcal{D}|}\big\{\bj_i - \lambda \|\bx-\bx_i\|\big\}
    \end{align*}
\end{definition}
It follows from the Lipschitz condition on $J(\cdot,\varepsilon)$ that $\jlb(\bx) \leq J(\bx,\varepsilon) \leq \jub(\bx)$ for all $\bx\in\mathbb{X}$, as long as $\lambda \geq L_J$.
Our key finding is that, under appropriate choice of $\lambda$, our policy is better than $\jub(\cdot)$, that is: $J^\pi(\bx) \leq \jub(\bx)$.

\begin{theorem}[Policy evaluation inequality]\label{thm:pi-eval}
    Assume we are in any of the conditions of Proposition \ref{prop:recursive-feasibility}, $\gamma L_f < 1$, and that the dataset $\mathcal{D}$ contains \emph{trajectories}, i.e., for any  $(\bx_i,\bu_i)\in\mathcal{D}$, there exists $j\leq|\mathcal{D}| : \bx_j=f(\bx_i,\bu_i)\in\mathcal{D}$. Then:
    \begin{equation}
    \lambda\geq\left(\frac{1+\gamma L_f}{1-\gamma L_f}\right)L_J \implies J^\pi(\bx) \leq \jub(\bx) \quad\forall \bx\in\mathbb{X}.\label{eq:lambda}
    \end{equation}
\end{theorem}
\begin{proof}
    The full proof is in Appendix \ref{proof:pi-eval}., here we provide a sketch. The key idea is establishing that
    $
    \left(\mathcal{T}^\pi \jub\right)(\bx) \leq \jub(\bx) \quad\forall \bx\in\mathbb{X},
    $
    where $\mathcal{T}^\pi$ is the Bellman operator under policy $\pi$ \citep[Ch. 1]{bertsekas2011dynamic}. Then, the monotonicity of this operator $\left(J_1 \preceq J_2 \implies \mathcal{T}^\pi J_1\preceq \mathcal{T}^\pi J_2 \right)$, together with the fact that its unique fixed point is $J^\pi$ yields the thesis.
\end{proof}
The preceding theorem establishes that our policy is at least as good as a conservative upper bound on the optimal cost-to-go. We wish to emphasize that $\jub(\cdot)$ is a data-dependent bound that \emph{improves} (i.e. becomes smaller) with more data, which is expected to improve the performance of policy $J^\pi$. Our work builds on \cite{castellano2025nonparametric}, where the authors study unconstrained optimization problems in the context of Reinforcement Learning, adapted here to constrained settings. We close this section by establishing performance guarantees for policy $\pi_{\mathcal{D}}$.
\begin{theorem}[Performance guarantees]\label{thm:performance-guarantees}
Let $\beta>0$ and $0<\eta\ll1$.
Assume policy $\pi_{\mathcal{D}}$ is recursively feasible and that the conditions of Theorem \ref{thm:pi-eval} are satisfied. If the dataset $\mathcal{D}$ has \emph{sufficient coverage}, in the sense that:
$\forall \bx\in\mathbb{X}~~\exists~ \bx_i \in\mathcal{D}: \|\bx-\bx_i\|\leq \frac{\beta}{(2+\beta)\lambda}\left(\bj_i+\eta\right)$
, then:
\begin{equation}    \sup_{\bx\in\mathbb{X}}~\frac{J^\pi(\bx) - J(\bx,\varepsilon)}{J(\bx,\varepsilon)+\eta} \leq \beta\;. \label{eq:relative-performance}
\end{equation}
Moreover, if $\eta>L\varepsilon$ then:
$
    \sup_{\bx\in\mathbb{X}}~\frac{J^\pi(\bx) - J(\bx)}{J(\bx)+\eta} \leq \frac{\beta\eta + L\varepsilon}{\eta - L\varepsilon}\;. \label{eq:true-relative-performance}
$
\end{theorem}
The proof is in Appendix \ref{app:proof-performance-guarantees}. We highlight that \eqref{eq:relative-performance} bounds the \emph{relative suboptimality} of our policy with respect to the optimal solution of the \emph{conservative} problem \eqref{eq:conservative-prob}. By contrast, the next equation quantifies the \emph{relative gap} to the optimal solution of the \emph{original} problem \eqref{eq:og-prob}. 
\revise{\paragraph{On bounding Lipschitz constants} Our nonparametric policy (and the upper- and lower-bounds of $J$) depends on a hyperparameter $\lambda > 0$. Sufficient conditions for a value of $\lambda$ that yields the results of theorems \ref{thm:pi-eval} and \ref{thm:performance-guarantees} in turn rely on the Lipschitz constants $L_f$ and $L_J$ (or at least upper bounds on them). We recognize this as a limitation of our work but also remark that these constants can be learned from trajectories, see e.g. \cite{knuth2021planning}.

}

\section{Algorithms}\label{sec:algorithms}
We now present two algorithms to drive the data-collection process for dataset $\mathcal{D}$. The first one (Algorithm \ref{alg:1}) assumes initial conditions can be sampled uniformly from $\mathbb{X}$. Offline MPC is then run from each sampled point to give an optimal trajectory for \eqref{eq:conservative-prob}. We establish PAC bounds \citep{haussler2018probably} for this algorithm next, since people are fond of them.

\begin{figure}[t]

\refstepcounter{algorithm}
\begin{algobox}{Data collector}\label{alg:1}

\LinesNumbered
\KwIn{Conservative threshold $\varepsilon>0$. Budget $N$.}
\textbf{Initialize:} $\mathcal{D}=\emptyset$\;

\For{$i=1,\ldots,N$}{
    Sample $\bx_i \sim \operatorname{Uniform}(\mathbb{X})$\;
    
    Get trajectory ${\tau}_i=\{\left(\bx_k,\bu_k,\bj_k\right)\}_{k=i}^{T+i-1}$ 
     solving \eqref{eq:conservative-prob} from $\bx_i$, add each transition to $\mathcal{D}.$\;
}

\KwOut{Nonparametric policy $\pi_{\mathcal{D}}$ with feasibility \& performance guarantees (Prop.~\ref{prop:sample-complexity})}

\end{algobox}

\end{figure}




\begin{proposition}[Sample complexity for Algorithm \ref{alg:1}]\label{prop:sample-complexity}
Let $\beta>0$, $\eta>L\varepsilon(1+\tfrac{1}{\beta})$ \revise{and $\lambda>0$ satisfying \eqref{eq:lambda}.} Pick any $\delta\in(0,1)$. Define
$r \triangleq \frac{1}{2}\min\left\{\frac{\eta}{\lambda}\cdot\frac{\beta\eta-L\varepsilon(1+\beta)}{(2+\beta)\eta-L\varepsilon(1+\beta)},\,\frac{\varepsilon}{L_f}\right\}.$ With probability at least $1-\delta$, Algorithm \ref{alg:1} outputs both a recursively feasible and $\beta$-optimal policy over $\mathbb{X}$ after at most:
    $$
    \mathcal{O}\left(\operatorname{N}_{\text{cover}}\left(\mathbb{X}; r\right)\log\operatorname{N}_{\text{cover}}\left(\mathbb{X}; r\right)\log\frac{1}{\delta}\right)
    $$
    iterations, where $\operatorname{N}_{\text{cover}}\left(\mathbb{X}; r\right)$ is the covering number\footnote{Formally defined as  $\min\left\{n>0 : \exists \bx_1,\ldots,\bx_n\in\mathbb{X}: \bigcup_{i\leq n} \mathbb{B}\left(\bx_i, r\right)\supseteq \mathbb{X}\right\}$} of $\mathbb{X}$ with balls of radius $r$, and ``$\beta$-optimal policy'' means $\sup_{\bx\in\mathbb{X}}~\frac{J^\pi(\bx) - J(\bx)}{J(\bx)+\eta} \leq \beta$, i.e., the gap with respect to the original problem \eqref{eq:og-prob}. 
    \label{prop:15}
\end{proposition}
See Appendix \ref{proof:4} for a proof of Proposition \ref{prop:sample-complexity}.
Algorithm \ref{alg:1} enjoys the guarantees above but may be inefficient due to stochastic sampling. \revise{An alternative verification procedure is provided in Algorithm \ref{alg:verification}, which recursively partitions $\mathbb{X}$ into disjoint cells until feasibility and a desired performance is guaranteed. A detailed description can be found in Appendix \ref{appen:c5}.}
\algrenewcommand{\algorithmiccomment}[1]{{\hfill\color{blue}\(\triangleright\) #1}}

\begin{figure}
\centering
\includegraphics[width=\linewidth]{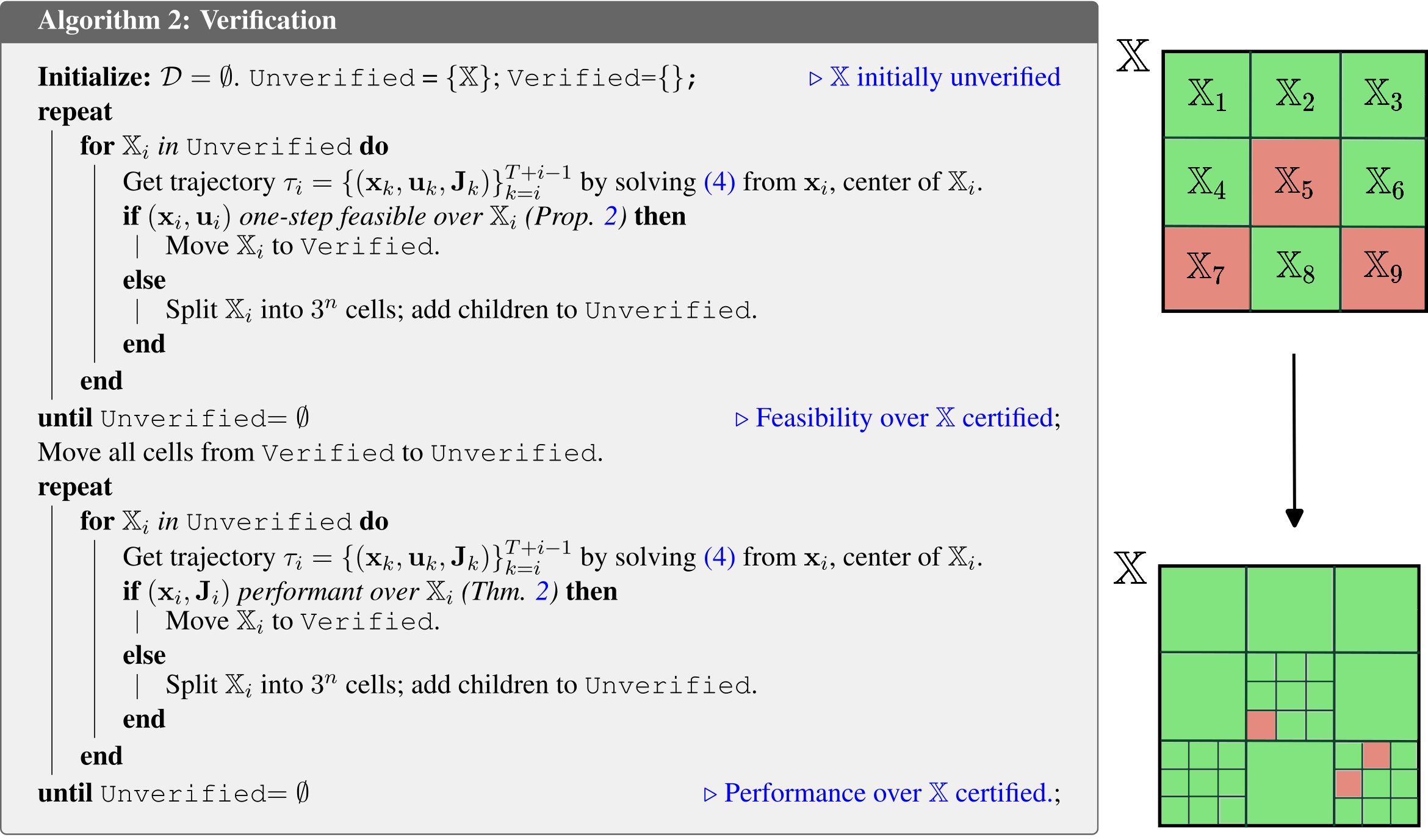}
\caption{Algorithm {\color{blue}2} and the visualization of the cell verification/splitting method. Each cell is tested against two criteria: (i) one-step feasibility and (ii) performance. 
Cells that pass are shown in {\color{green!70!black}{green}} and kept in \texttt{Verified}; those that fail are shown in {\color{red}{red}} and are split into $3^n$ children cells. The children are then re-verified using the same criteria, proceeding recursively.}
\label{alg:verification}
\end{figure}
\RestyleAlgo{boxed}
\RestyleAlgo{ruled}
\LinesNumbered              

\section{Experiments}\label{sec:experiments}
The purpose of the following numerical simulations are two-fold. First, we want to establish a trade-off between the speedup of our method and its performance gap. As one would expect, more data will improve performance but at the cost of higher latency. Second, we show Algorithm \ref{alg:verification} in action, recursively verifying a domain $\mathbb{X}$ by the cell-splitting method. 
\paragraph*{On ``accelerating'' MPC:}
One potential bottleneck of implementing our policy $\pi_{\mathcal{D}}$ is that inference requires querying distances $\|\bx-\bx_i\|$ of a test point $\bx$ to each point $\bx_i$ in the dataset. We use \texttt{FAISS} \citep{johnson2019billion, douze2024faiss}, a GPU-enabled library for fast similarity search, allowing us to enjoy a substantial speed-up with respect to standard MPC\revise{. We run \texttt{FAISS} in CPU-only mode. Instead of fully solving  $\iota(\bx) = \argmin_{1\leq i \leq |\mathcal{D}|}\big\{\bj_i + \lambda \cdot \|\bx-\bx_i\|\big\}$, we invoke \texttt{FAISS} to get the $k$-nearest-neighbors to a query point $\bx$. Call these points $\mathcal{N}_k(\bx)$. Then, the minimization done at each step is approximated by: 
$
\iota(\bx) \approx \argmin_{i:~\bx_i\in\mathcal{N}_k(\bx)}\big\{\bj_i + \lambda \cdot \|\bx-\bx_i\|\big\}\;,
\label{eq:approximate-index}
$
which amounts to taking the smallest value of a $k$-dimensional array. A large enough value of $k$ will ensure that $\mathcal{N}_k(\bx)$ contains the true minimizer, but renders the look-up slower.  We use $k=100$. 


}
\revise{\paragraph*{Hardware/Software Setup} We use \texttt{do-mpc}~\citep{fiedler2023mpc}, an open-source python module for MPC to conduct the experiments. This software interfaces with \texttt{CasADi}~\citep{casadi2019}, which uses a symbolic framework to build each problem, and runs automatic differentiation to obtain trajectories and controls. The optimization solver running in the background is \texttt{IPOPT}~\citep{wachter2006implementation, biegler2009large}, with a tolerance of $10^{-8}$. All experiments are run on a Macbook Air M2 with 8GB of RAM, without GPUs.}

\paragraph*{Trade-offs between latency and performance:} \label{sec:experiments-tradeoff}
We study trade-offs between latency (small computation time) and performance (small optimality gap) on two different benchmarks: an \textbf{inverted pendulum}, in which we wish to stabilize the pendulum near the unstable equilibrium, and a \textbf{minimum time}, unstable LTI system that we wish to drive to the origin with penalized control effort. Due to space limitations, we relegate the details of these environments to Appendix \ref{app:experiments}. For these two environments, we consider different datasets $\mathcal{D}$ obtained by uniformly partitioning the state space into a grid, with $g\in\{3,5,7,9,11\}$ grid points per dimension. Then, \eqref{eq:conservative-prob} is solved for each point in the grid, and horizon $T=100$ trajectories are added to the dataset. The performance of both MPC controllers (with varying lookahead horizon $H\leq T$) and our nonparametric policies (labeled \texttt{MINT}, for \texttt{M}onotonically \texttt{I}mproving bound-based \texttt{N}onparametric policy based on \texttt{T}rajectories) are evaluated on $M=100$ trajectories from random initial conditions. Figure \ref{fig:pendulum-and-min-time}, shows the result of these trajectories. Boxes show the interquartile range (25\% - 75\%) of the empirical distribution over the $M$ trajectories; black lines correspond to the median values, means are shown in green, whiskers extend to 1.5 times the inter-quartile range. Outliers are shown as black circles. The left panel shows the time taken to get one control action (in $ms$), the middle one shows the empirical normalized gap $\tfrac{J^\pi(\bx)-J(\bx)}{J(\bx)}$ for the different controllers. The right panels shows the trade-off between the median normalized gap and the median time taken to get controls. Our nonparametric policy strikes a good balance between good performance (low gap) with lower computation time.

\begin{figure}[t]
    \centering
    \includegraphics[width=\linewidth]{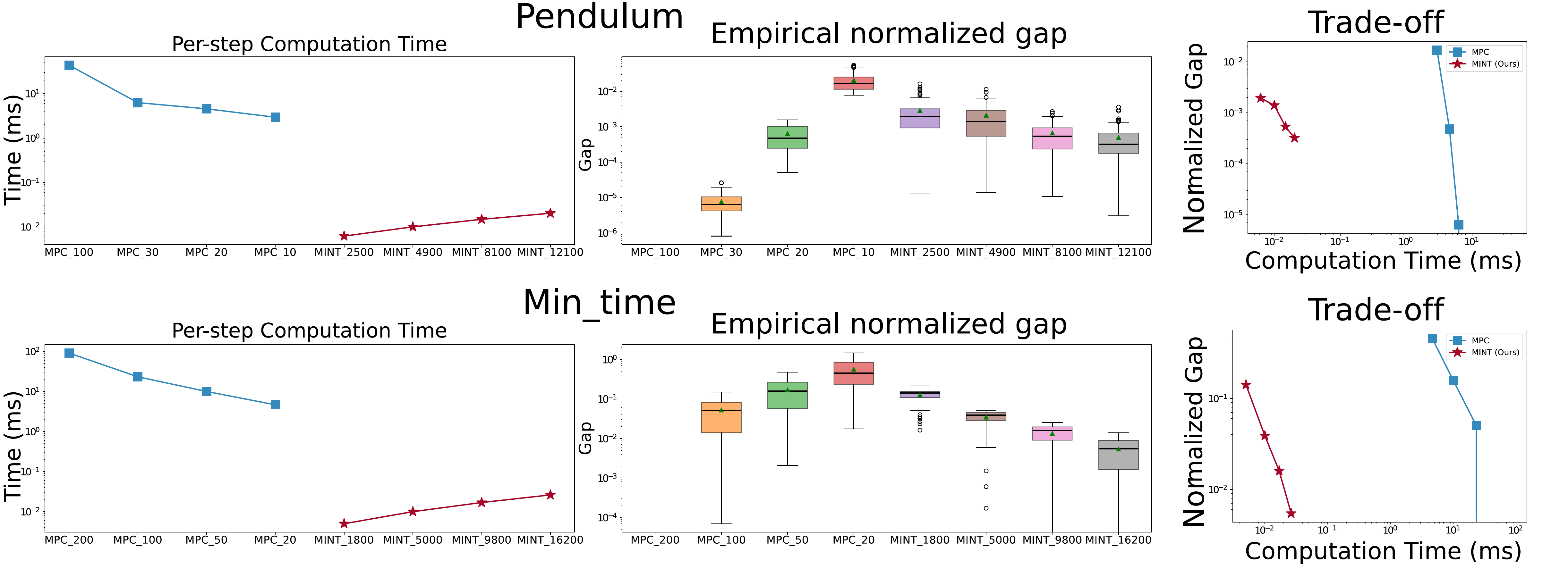}
    \caption{Statistics for the inverted pendulum (top) and minimum time problem (bottom) over 100 trajectories. Left: Per-step latency (in $ms$) for each controller. Our controllers (\texttt{MINT\_{XXXX}}, {\color{red} red $\star$'s}) are ordered left to right from smallest to largest dataset $\mathcal{D}$, MPC controllers ({\color{blue} blue $\square$'s}) ordered from largest to smallest horizon. Middle: Distribution of the relative optimality gap. Boxes correspond to the interquartile range $(25\% - 75\%)$, black line shows the median and the green arrow corresponds to the mean. Right: Trade-off between computation time and relative gap for our the controllers. Our method is substantially faster than MPC and, with sufficient data, outperforms MPC with shorter lookahead horizons.
    }
    \label{fig:pendulum-and-min-time}
\end{figure}

\paragraph*{Verification}\label{sec:experiments-verification}
We test Algorithm \ref{alg:verification} on a discrete LQR problem. Details on the setup and on the hyperparameters of our method are in Appendix \ref{app:experiments}.  Figure \ref{fig:feasibility-experiments} shows Algorithm \ref{alg:verification} in action: it recursively partitions the state space into smaller cells until feasibility can be verified (top row of Fig. \ref{fig:feasibility-experiments}), then it focuses on guaranteeing optimality for each cell---with more splitting if needed (bottom row of Fig. \ref{fig:feasibility-experiments}). Feasibility is harder to verify near the boundary of $\mathbb{X}$ (requiring more splitting), to be expected from our guarantees on a feasibility radius (Prop. \ref{prop:local-feasibility}). A target gap is harder to verify near the origin, since our target gap \eqref{eq:relative-performance} is relative and $J(\bx) \approx 0$ in that vicinity.

\begin{figure}[ht]
\centering
\includegraphics[width=.98\linewidth]{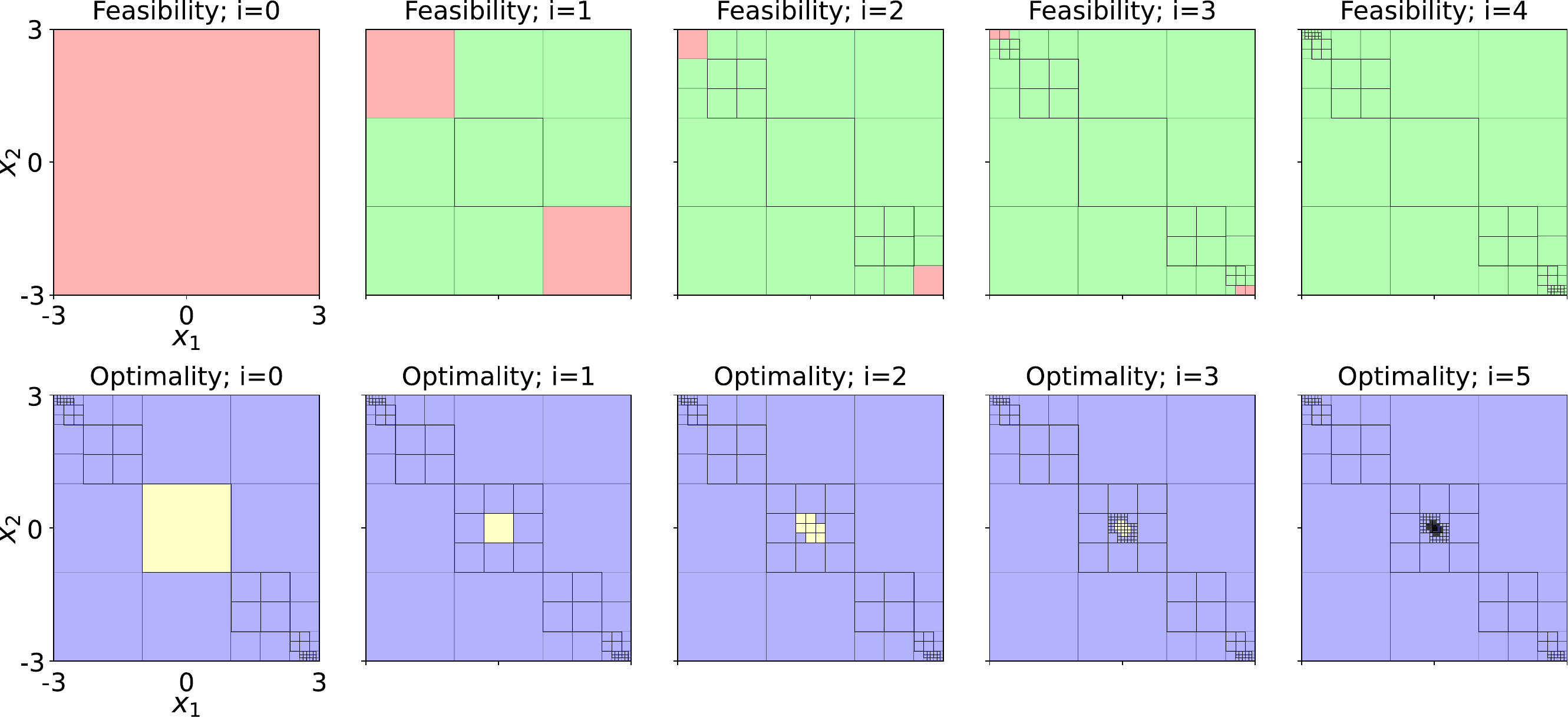}    
    \caption{Algorithm \ref{alg:verification} in action: feasibility (top row) and optimality (bottom row) certificates for LQR, iterations ordered left to right. Top: Shown in {\color{red} red} are cells that don't satisfy the feasibility condition (Prop. \ref{prop:local-feasibility}). The algorithm recursively splits each cell until verification, shown in {\color{green!70!black} green}. Bottom: verification of gap (Theorem \ref{thm:performance-guarantees}). Guaranteed suboptimal cells are shown in {\color{blue}blue}. At termination the algorithm has verified the whole state space $\mathbb{X}$.}
    \label{fig:feasibility-experiments}
\end{figure}

\section{Conclusions \& Future work}
We proposed a nonparametric, data-driven scheme to accelerate MPC by reusing offline-computed trajectories. Our policy picks stored controls by trading off cost-to-go and state proximity. Under mild Lipschitz and coverage assumptions, the controller enjoys recursive feasibility and has explicit performance bounds; with sufficient dataset coverage the relative suboptimality can be made arbitrarily small, establishing an explicit trade-off between desired performance and memory requirements. Inference is extremely fast and we achieve a few orders of magnitude of speed-up with respect to standard MPC. Future work includes testing and scalable implementations for high-dimensional systems and embedded control applications.



\acks{A.C. is grateful to Pedro Izquierdo Lehmann for insightful discussions and suggestions during the first stages of this manuscript. This work was supported by the NSF Global Centers program under Grant No.~2330450 and by the DOE Office of Science (ASCR) under Award No.~826565.}
\clearpage
\bibliography{refs}

@inproceedings{chen2018approximating,
  title={Approximating explicit model predictive control using constrained neural networks},
  author={Chen, Steven and Saulnier, Kelsey and Atanasov, Nikolay and Lee, Daniel D and Kumar, Vijay and Pappas, George J and Morari, Manfred},
  booktitle={2018 Annual American control conference (ACC)},
  pages={1520--1527},
  year={2018},
  organization={IEEE}
}

@article{blanchini1999set,
  title={Set invariance in control},
  author={Blanchini, Franco},
  journal={Automatica},
  volume={35},
  number={11},
  pages={1747--1767},
  year={1999},
  publisher={Elsevier}
}

@article{munos2011optimistic,
  title={Optimistic optimization of a deterministic function without the knowledge of its smoothness},
  author={Munos, R{\'e}mi},
  journal={Advances in neural information processing systems},
  volume={24},
  year={2011}
}

@inproceedings{garivier2016optimal,
  title={Optimal best arm identification with fixed confidence},
  author={Garivier, Aur{\'e}lien and Kaufmann, Emilie},
  booktitle={Conference on Learning Theory},
  pages={998--1027},
  year={2016},
  organization={PMLR}
}

@inproceedings{cutler1980dynamic,
  title={Dynamic matrix control},
  author={Cutler, CR and Ramaker, BL},
  booktitle={A computer control algorithm. In joint automatic control conference},
  volume={17},
  pages={72},
  year={1980}
}

@article{richalet1978model,
  title={Model predictive heuristic control},
  author={Richalet, Jacques and Rault, Andr{\'e} and Testud, JL and Papon, J},
  journal={Automatica (journal of IFAC)},
  volume={14},
  number={5},
  pages={413--428},
  year={1978},
  publisher={Pergamon Press, Inc. Elmsford, NY, USA}
}

@article{drgovna2020all,
  title={All you need to know about model predictive control for buildings},
  author={Drgo{\v{n}}a, J{\'a}n and Arroyo, Javier and Figueroa, Iago Cupeiro and Blum, David and Arendt, Krzysztof and Kim, Donghun and Oll{\'e}, Enric Perarnau and Oravec, Juraj and Wetter, Michael and Vrabie, Draguna L and others},
  journal={Annual reviews in control},
  volume={50},
  pages={190--232},
  year={2020},
  publisher={Elsevier}
}

@incollection{alessio2009survey,
  title={A survey on explicit model predictive control},
  author={Alessio, Alessandro and Bemporad, Alberto},
  booktitle={Nonlinear model predictive control: towards new challenging applications},
  pages={345--369},
  year={2009},
  publisher={Springer}
}

@inproceedings{di2018real,
  title={Real-time optimization and model predictive control for aerospace and automotive applications},
  author={Di Cairano, Stefano and Kolmanovsky, Ilya V},
  booktitle={2018 annual American control conference (ACC)},
  pages={2392--2409},
  year={2018},
  organization={IEEE}
}

@inproceedings{hrovat2012development,
  title={The development of model predictive control in automotive industry: A survey},
  author={Hrovat, Davor and Di Cairano, Stefano and Tseng, H Eric and Kolmanovsky, Ilya V},
  booktitle={2012 IEEE International Conference on Control Applications},
  pages={295--302},
  year={2012},
  organization={IEEE}
}

@article{forbes2015model,
  title={Model predictive control in industry: Challenges and opportunities},
  author={Forbes, Michael G and Patwardhan, Rohit S and Hamadah, Hamza and Gopaluni, R Bhushan},
  journal={IFAC-PapersOnLine},
  volume={48},
  number={8},
  pages={531--538},
  year={2015},
  publisher={Elsevier}
}

@article{mattingley2011receding,
  title={Receding horizon control},
  author={Mattingley, Jacob and Wang, Yang and Boyd, Stephen},
  journal={IEEE Control Systems Magazine},
  volume={31},
  number={3},
  pages={52--65},
  year={2011},
  publisher={IEEE}
}

@article{rawlings2002stability,
  title={The stability of constrained receding horizon control},
  author={Rawlings, James B and Muske, Kenneth R},
  journal={IEEE transactions on automatic control},
  volume={38},
  number={10},
  pages={1512--1516},
  year={2002},
  publisher={IEEE}
}

@article{garcia1989model,
  title={Model predictive control: Theory and practice—A survey},
  author={Garcia, Carlos E and Prett, David M and Morari, Manfred},
  journal={Automatica},
  volume={25},
  number={3},
  pages={335--348},
  year={1989},
  publisher={Elsevier}
}

@article{hertneck2018learning,
  title={Learning an approximate model predictive controller with guarantees},
  author={Hertneck, Michael and K{\"o}hler, Johannes and Trimpe, Sebastian and Allg{\"o}wer, Frank},
  journal={IEEE Control Systems Letters},
  volume={2},
  number={3},
  pages={543--548},
  year={2018},
  publisher={IEEE}
}

@article{bemporad2002explicit,
  title={The explicit linear quadratic regulator for constrained systems},
  author={Bemporad, Alberto and Morari, Manfred and Dua, Vivek and Pistikopoulos, Efstratios N},
  journal={Automatica},
  volume={38},
  number={1},
  pages={3--20},
  year={2002},
  publisher={Elsevier}
}

@inproceedings{
castellano2025nonparametric,
title={Nonparametric Policy Improvement in Continuous Action Spaces via Expert Demonstrations},
author={Agustin Castellano and Sohrab Rezaei and Jared Markowitz and Enrique Mallada},
booktitle={Reinforcement Learning Conference},
year={2025}}

@article{jadbabaie2002unconstrained,
  title={Unconstrained receding-horizon control of nonlinear systems},
  author={Jadbabaie, Ali and Yu, Jie and Hauser, John},
  journal={IEEE Transactions on Automatic Control},
  volume={46},
  number={5},
  pages={776--783},
  year={2002},
  publisher={IEEE}
}

@article{bellman1954theory,
  title={The theory of dynamic programming},
  author={Bellman, Richard},
  journal={Bulletin of the American Mathematical Society},
  volume={60},
  number={6},
  pages={503--515},
  year={1954}
}

@book{bertsekas2012dynamic,
  title={Dynamic programming and optimal control: Volume I},
  author={Bertsekas, Dimitri},
  volume={4},
  year={2012},
  publisher={Athena scientific}
}

@misc{hastie2009elements,
  title={The elements of statistical learning},
  author={Hastie, Trevor and Tibshirani, Robert and Friedman, Jerome and others},
  year={2009},
  publisher={Springer series in statistics New-York}
}

@article{johnson2019billion,
  title={Billion-scale similarity search with {GPUs}},
  author={Johnson, Jeff and Douze, Matthijs and J{\'e}gou, Herv{\'e}},
  journal={IEEE Transactions on Big Data},
  volume={7},
  number={3},
  pages={535--547},
  year={2019},
  publisher={IEEE}
}

@article{douze2024faiss,
      title={The Faiss library},
      author={Matthijs Douze and Alexandr Guzhva and Chengqi Deng and Jeff Johnson and Gergely Szilvasy and Pierre-Emmanuel Mazaré and Maria Lomeli and Lucas Hosseini and Hervé Jégou},
      year={2024},
      eprint={2401.08281},
      journal={arXiv},
      archivePrefix={arXiv},
      primaryClass={cs.LG}
}

@article{adusumilli2022neyman,
  title={Neyman allocation is minimax optimal for best arm identification with two arms},
  author={Adusumilli, Karun},
  journal={arXiv preprint arXiv:2204.05527},
  year={2022}
}

@article{siegelmann2025stability,
  title={Stability Analysis and Data-driven Verification via Recurrent Lyapunov Functions},
  author={Siegelmann, Roy and Shen, Yue and Paganini, Fernando and Mallada, Enrique},
  journal={IEEE Transactions on Automatic Control},
  volume={7},
  year={2025}
}

@article{bertsekas2011dynamic,
  title={Dynamic programming and optimal control 3rd edition, volume ii},
  author={Bertsekas, Dimitri P and others},
  journal={Belmont, MA: Athena Scientific},
  volume={1},
  year={2011}
}

@inproceedings{fogliato2024framework,
  title={A Framework for Efficient Model Evaluation through Stratification, Sampling, and Estimation},
  author={Fogliato, Riccardo and Patil, Pratik and Monfort, Mathew and Perona, Pietro},
  booktitle={European Conference on Computer Vision},
  pages={140--158},
  year={2024},
  organization={Springer}
}

@article{sinclair2020adaptive,
  title={Adaptive discretization for model-based reinforcement learning},
  author={Sinclair, Sean and Wang, Tianyu and Jain, Gauri and Banerjee, Siddhartha and Yu, Christina},
  journal={Advances in Neural Information Processing Systems},
  volume={33},
  pages={3858--3871},
  year={2020}
}

@article{mayne2014model,
  title={Model predictive control: Recent developments and future promise},
  author={Mayne, David Q},
  journal={Automatica},
  volume={50},
  number={12},
  pages={2967--2986},
  year={2014},
  publisher={Elsevier}
}

@article{mayne2000constrained,
  title={Constrained model predictive control: Stability and optimality},
  author={Mayne, David Q and Rawlings, James B and Rao, Christopher V and Scokaert, Pierre OM},
  journal={Automatica},
  volume={36},
  number={6},
  pages={789--814},
  year={2000},
  publisher={Elsevier}
}

@article{worthmann2012estimates,
  title={Estimates of the prediction horizon length in MPC: A numerical case study},
  author={Worthmann, Karl},
  journal={IFAC Proceedings Volumes},
  volume={45},
  number={17},
  pages={232--237},
  year={2012},
  publisher={Elsevier}
}

@article{grune2010analysis,
  title={Analysis of unconstrained nonlinear MPC schemes with time varying control horizon},
  author={Gr{\"u}ne, Lars and Pannek, J{\"u}rgen and Seehafer, Martin and Worthmann, Karl},
  journal={SIAM Journal on Control and Optimization},
  volume={48},
  number={8},
  pages={4938--4962},
  year={2010},
  publisher={SIAM}
}

@article{reble2012unconstrained,
  title={Unconstrained model predictive control and suboptimality estimates for nonlinear continuous-time systems},
  author={Reble, Marcus and Allg{\"o}wer, Frank},
  journal={Automatica},
  volume={48},
  number={8},
  pages={1812--1817},
  year={2012},
  publisher={Elsevier}
}

@book{blanchini2008set,
  title={Set-theoretic methods in control},
  author={Blanchini, Franco and Miani, Stefano and others},
  volume={78},
  year={2008},
  publisher={Springer}
}

@article{maddalena2020neural,
  title={A neural network architecture to learn explicit MPC controllers from data},
  author={Maddalena, Emilio Tanowe and Moraes, CG da S and Waltrich, Gierri and Jones, Colin N},
  journal={IFAC-PapersOnLine},
  volume={53},
  number={2},
  pages={11362--11367},
  year={2020},
  publisher={Elsevier}
}

@article{chen1998quasi,
  title={A quasi-infinite horizon nonlinear model predictive control scheme with guaranteed stability},
  author={Chen, Hong and Allg{\"o}wer, Frank},
  journal={Automatica},
  volume={34},
  number={10},
  pages={1205--1217},
  year={1998},
  publisher={Elsevier}
}

@article{lofberg2012oops,
  title={Oops! I cannot do it again: Testing for recursive feasibility in MPC},
  author={L{\"o}fberg, Johan},
  journal={Automatica},
  volume={48},
  number={3},
  pages={550--555},
  year={2012},
  publisher={Elsevier}
}

@article{kuzborskij2020locally,
  title={Locally-adaptive nonparametric online learning},
  author={Kuzborskij, Ilja and Cesa-Bianchi, Nicolo},
  journal={Advances in Neural Information Processing Systems},
  volume={33},
  pages={1679--1689},
  year={2020}
}

@article{bucsoniu2018continuous,
  title={Continuous-action planning for discounted infinite-horizon nonlinear optimal control with Lipschitz values},
  author={Bu{\c{s}}oniu, Lucian and P{\'a}ll, El{\H{o}}d and Munos, R{\'e}mi},
  journal={Automatica},
  volume={92},
  pages={100--108},
  year={2018},
  publisher={Elsevier}
}

@article{bonnans1998optimization,
  title={Optimization problems with perturbations: A guided tour},
  author={Bonnans, J Fr{\'e}d{\'e}ric and Shapiro, Alexander},
  journal={SIAM review},
  volume={40},
  number={2},
  pages={228--264},
  year={1998},
  publisher={SIAM}
}

@book{rockafellar1998variational,
  title={Variational analysis},
  author={Rockafellar, R Tyrrell and Wets, Roger JB},
  year={1998},
  publisher={Springer}
}

@article{sibai2026recurrence,
  title={Recurrence of nonlinear control systems: Entropy, bit rates, and finite alphabet controllers},
  author={Sibai, Hussein and Mallada, Enrique},
  journal={Nonlinear Analysis: Hybrid Systems},
  volume={59},
  pages={101649},
  year={2026},
  publisher={Elsevier}
}

@inproceedings{shen2022model,
  title={Model-free learning of regions of attraction via recurrent sets},
  author={Shen, Yue and Bichuch, Maxim and Mallada, Enrique},
  booktitle={2022 IEEE 61st Conference on Decision and Control (CDC)},
  pages={4714--4719},
  year={2022},
  organization={IEEE}
}

@Article{casadi2019,
  author = {Joel A E Andersson and Joris Gillis and Greg Horn
            and James B Rawlings and Moritz Diehl},
  title = {{CasADi} -- {A} software framework for nonlinear optimization
           and optimal control},
  journal = {Mathematical Programming Computation},
  volume = {11},
  number = {1},
  pages = {1--36},
  year = {2019},
  publisher = {Springer},
  doi = {10.1007/s12532-018-0139-4}
}

@article{fiedler2023mpc,
  title={do-mpc: Towards FAIR nonlinear and robust model predictive control},
  author={Fiedler, Felix and Karg, Benjamin and L{\"u}ken, Lukas and Brandner, Dean and Heinlein, Moritz and Brabender, Felix and Lucia, Sergio},
  journal={Control Engineering Practice},
  volume={140},
  pages={105676},
  year={2023},
  publisher={Elsevier}
}

@article{wachter2006implementation,
  title={On the implementation of an interior-point filter line-search algorithm for large-scale nonlinear programming},
  author={W{\"a}chter, Andreas and Biegler, Lorenz T},
  journal={Mathematical programming},
  volume={106},
  number={1},
  pages={25--57},
  year={2006},
  publisher={Springer}
}

@article{biegler2009large,
  title={Large-scale nonlinear programming using IPOPT: An integrating framework for enterprise-wide dynamic optimization},
  author={Biegler, Lorenz T and Zavala, Victor M},
  journal={Computers \& Chemical Engineering},
  volume={33},
  number={3},
  pages={575--582},
  year={2009},
  publisher={Elsevier}
}

@article{knuth2021planning,
  title={Planning with learned dynamics: Probabilistic guarantees on safety and reachability via lipschitz constants},
  author={Knuth, Craig and Chou, Glen and Ozay, Necmiye and Berenson, Dmitry},
  journal={IEEE Robotics and Automation Letters},
  volume={6},
  number={3},
  pages={5129--5136},
  year={2021},
  publisher={IEEE}
}

@article{carnerero2023kernel,
  title={Kernel-based state-space kriging for predictive control},
  author={Carnerero, A Daniel and Ramirez, Daniel R and Limon, Daniel and Alamo, Teodoro},
  journal={IEEE/CAA Journal of Automatica Sinica},
  volume={10},
  number={5},
  pages={1263--1275},
  year={2023},
  publisher={IEEE}
}

@article{huang2023robust,
  title={Robust and kernelized data-enabled predictive control for nonlinear systems},
  author={Huang, Linbin and Lygeros, John and D{\"o}rfler, Florian},
  journal={IEEE Transactions on Control Systems Technology},
  volume={32},
  number={2},
  pages={611--624},
  year={2023},
  publisher={IEEE}
}

@article{ganguly2025explicit,
  title={Explicit feedback synthesis driven by quasi-interpolation for nonlinear model predictive control},
  author={Ganguly, Siddhartha and Chatterjee, Debasish},
  journal={IEEE Transactions on Automatic Control},
  volume={70},
  number={7},
  pages={4751--4758},
  year={2025},
  publisher={IEEE}
}

@inproceedings{bayer2016tube,
  title={A tube-based approach to nonlinear explicit MPC},
  author={Bayer, Florian A and Brunner, Florian D and Lazar, Mircea and Wijnand, Marc and Allg{\"o}wer, Frank},
  booktitle={2016 IEEE 55th Conference on Decision and Control (CDC)},
  pages={4059--4064},
  year={2016},
  organization={IEEE}
}

@article{trinh2016explicit,
  title={Explicit model predictive control via nonlinear piecewise approximations},
  author={Trinh, Van-Vuong and Alamir, Mazen and Bonnay, Patrick and Bonne, Fran{\c{c}}ois},
  journal={IFAC-PapersOnLine},
  volume={49},
  number={18},
  pages={259--264},
  year={2016},
  publisher={Elsevier}
}

@article{canale2009set,
  title={Set membership approximation theory for fast implementation of model predictive control laws},
  author={Canale, Massimo and Fagiano, Lorenzo and Milanese, Mario},
  journal={Automatica},
  volume={45},
  number={1},
  pages={45--54},
  year={2009},
  publisher={Elsevier}
}

@article{haussler2018probably,
  title={The probably approximately correct (PAC) and other learning models},
  author={Haussler, David and Warmuth, Manfred},
  journal={The Mathematics of Generalization},
  pages={17--36},
  year={2018},
  publisher={CRC Press}
}

\clearpage
\appendix
\section{}
\revise{In what follows we provide some useful inequalities that will aid the proofs of some results in the paper. First, by the Lipschitz assumption on $f(\bx,\bu)$ and the fact that the origin is an equilibrium point, the following holds.:
}


For all $\bx\in\mathbb{X}$ and for all $\bu\in\mathbb{U}$ we have:
\begin{equation}
    \left\|f(\bx,\bu)\right\| \leq L_f\|\bx\|+L_u\|\bu\|\;.\label{eq:bound-dist-to-origin}
\end{equation}    
\begin{remark}[On the need for global Lipschitz constants]
    For the previous result to make sense for any $\bx\in\mathbb{X},\bu\in\mathbb{U}$ $L_f$ and $L_u$ must be global Lipschitz constants, since we are comparing $(\bx,\bu)$ to $(\mathbf{0},\mathbf{0})$.
\end{remark}

\subsection*{Bounding trajectories}
We use $\bx_t=\phi(\bx_0,\bu_{0:t-1})$ to denote the solution at time $t$ starting from $\bx_0$, under control law $\bu_{0:t-1}=[\bu_0, \bu_1,\ldots\bu_{t-1}]$. Similarly, let $\bx'_t=\phi(\bx_0',\bu_{0:t-1}')$. We have the following result:
\begin{equation}
    \left\|\bx_t - \bx_t'\right\| \leq L_f^k \cdot \left\|\bx_0 - \bx_0'\right\| + L_u\sum_{\ell=0}^{t-1}L_f^{t-1-\ell}\left\|\bu_\ell - \bu_\ell'\right\|\quad \forall t\geq 0~.\label{eq:bread-and-butter}
\end{equation}
In particular, for two different states $\bx_0, \bx_0'$ under the same control $\bu_0$:
\begin{equation}
    \left\|\bx_1-\bx_1'\right\| \leq L_f\left\|\bx_0-\bx_0'\right\|\;.
\end{equation}
\section{Proofs}
\subsection{Proposition \ref{prop:local-feasibility}\label{proof:one-step-feasibility}} 

Let $\left(\bx, \bu, \bx'\right)$ be a triplet with $\bx'=f(\bx,\bu)$. We want to study whether applying the same control $\bu$ for a different state $\bx_0$ (close to $\bx$) is feasible.\\
We define the \emph{radius} of the feasible set $\mathbb{X}$ as:
\begin{equation}
    R \triangleq \sup\left\{r\geq 0: \mathbb{B}_{\mathbb{X}}(\mathbf{0}, r) \subseteq \mathbb{X}\right\}\;.
\end{equation}
We will show the following two inequalities (the latter is Proposition \ref{prop:local-feasibility}).
\begin{enumerate}
    \item If $(\bx, \bu)$ is feasible, then $(\bx_0,\bu)$ is feasible for all $\bx_0 \in \mathbb{B}\big(\bx, r(\bx,\bu)\big)$, where:
    \begin{equation}
        r(\bx,\bu) \triangleq \max\left\{0, \frac{R-L_u\|\bu\|}{L_f}-\|\bx\|\right\}
    \end{equation}
    \item If $(\bx,\bu)$ is feasible and $\bx'=f(\bx,\bu)$, then $(\bx_0,\bu)$ is feasible for all $\bx_0 \in \mathbb{B}\big(\bx, r(\bx')\big)$, where:
    \begin{equation}
        r(\bx') \triangleq \frac{R-\|\bx'\|}{L_f} \leq \frac{\operatorname{dist}(\bx',\partial\mathbb{X})}{L_f}
    \end{equation}
\end{enumerate}

\begin{proof}
\begin{enumerate}
    \item Let $\bx'=f(\bx,\bu)$. By virtue of \eqref{eq:bound-dist-to-origin}:
    \begin{equation}
        \|\bx'\| \leq L_f\|\bx\|+L_u\|\bu\|\label{eq:pf-1step-0}
    \end{equation}
    Let $\bx_0\in\mathbb{X}: \|\bx_0-\bx\|\leq r$, and $\bx_0' = f(\bx_0,\bu)$ be the successor state under the same $\bu$. We have:
    \begin{align}
        \|\bx_0'\|-\|\bx'\| &\leq \|\bx_0'-\bx'\| \leq L_f r \implies \label{eq:pf-1step-1}\\
        \|\bx_0'\| &\leq \|\bx'\| + L_f r \leq L_f\left(\|\bx\|+r\right)+L_u\|\bu\|,\label{eq:pf-1step-2}
    \end{align}

where the first inequality in \eqref{eq:pf-1step-1} follows from the reverse triangle inequality, and the second one from \eqref{eq:bread-and-butter} (for two successor states under same control $\bu$). In \eqref{eq:pf-1step-2} we rearrange terms and use \eqref{eq:pf-1step-0}.

Imposing the right hand side of \eqref{eq:pf-1step-2} be smaller than $R$:
\begin{equation}
    L_f\left(\|\bx\|+r\right)+L_u\|\bu\| \leq R \implies r \leq \frac{R-L_u\|\bu\|}{L_f}-\|\bx\|\;,
\end{equation}
as desired. 
\item Re-using the first inequality in \eqref{eq:pf-1step-2}, and imposing that it be upper bounded by $R$:
\begin{equation}
\|\bx_0'\|\leq \|\bx'\|+L_f r \leq R \implies r\leq\frac{R-\|\bx'\|}{L_f}\;.
\end{equation}
For the remaining inequality, note:
\begin{equation}
    \operatorname{dist}\left(\bx',\partial\mathbb{X}\right) \geq  R-\|\bx'\|\quad\forall \bx'\in\mathbb{X}\;.
\end{equation}
\end{enumerate}
\end{proof}



\subsection{Theorem \ref{thm:pi-eval}} \label{proof:pi-eval}
\begin{proof}
Notice that, by assumption, we are in the conditions of Proposition \ref{prop:recursive-feasibility}. This implies policy $\pi_{\mathcal{D}}$ is feasible, and hence:
$$
J^\pi(\bx) < +\infty \quad \forall \bx\in\mathbb{X}.$$
Let $\mathcal{T}^\pi:\mathcal{J}\to\mathcal{J}$ be the Bellman operator \citep[Ch. 1]{bertsekas2011dynamic} of policy $\pi$, mapping costs-to-go $J\in\mathcal{J}$ onto $\mathcal{J}$, defined by:
$$
\left(T^\pi J\right)(\bx) = c(\bx,\pi(\bx))+\gamma \cdot J\big(f(\bx,\pi(\bx))\big)
$$
The following are two well-known facts of $\mathcal{T}^\pi$\citep[Lemma 1.1.1; Prop 1.2.1]{bertsekas2011dynamic}:
\begin{enumerate}
    \item For any policy $\pi$, $\mathcal{T}^\pi$ is monotone, i.e.
    $$J_1(\bx) \leq J_2(\bx) \quad\forall \bx \implies \left(\mathcal{T}^\pi J_1\right)(\bx) \leq \left(\mathcal{T}^\pi J_2\right)(\bx)\quad \forall \bx $$
    \item $J^\pi$ is the unique fixed point of $\mathcal{T}^\pi$:
    $$
    \lim_{k\to\infty}\overbrace{\left(\mathcal{T}^\pi\circ\mathcal{T}^\pi\circ\ldots \mathcal{T}^\pi\right)}^{k~\text{times}}J = J^\pi\quad\forall J\in\mathcal{J}.
    $$
\end{enumerate}
The combination of these two facts leads to the following lemma:
\begin{lemma}
    If $J:\mathbb{X}\to\mathbb{R}$  satisfies:
    $$
    \left(T^\pi J\right)(\bx) \leq J(\bx)\quad\forall \bx\in\mathbb{X},
    $$
    then: 
    $$
    J^\pi(\bx)\leq J(\bx) \quad\forall \bx\in\mathbb{X}\;.
    $$
\end{lemma}
Then, to prove Theorem \ref{thm:pi-eval}, all that remains is  to show $\jub$ satisfies the hypothesis of the lemma above.\\
Fix $\bx\in\mathbb{X}$ and let
$$
\bu_i = \pi(\bx),
$$
where $i$ is the solution of: 
$$
\argmin_{1\leq i \leq |\mathcal{D}|}\big\{\bj_i + \lambda \cdot \|\bx-\bx_i\|\big\}\;.
$$
Define as well:
$$
\bx_i'=f(\bx_i, \bu_i),\quad \bx'=f(\bx,\bu_i),
$$
and the action-value function \citep{bertsekas2011dynamic} $Q:\mathbb{X}\times\mathbb{U}\to\mathbb{R}\cup\{+\infty\}$:
\begin{align*}
    Q(\bx_0, \bu_0) \triangleq  \min_{\bu_{1:T-1}}&\sum_{t=0}^{T-1}\gamma^t c(\bx_t, \bu_t) + F(\bx_T)\\ 
    \text{subject to:}~~& \bx_{t+1}=f(\bx_t, \bu_t),& t=0,\ldots,T-1\\
    &\bx_t\in\mathbb{X}_{-\varepsilon}, & t=1,\ldots,T-1\\
    & \bu_t\in\mathbb{U}, & t=0,\ldots,T-1
\end{align*}
which is the optimal cost-to-go of the original problem fixing both $\bx_0$ and $\bu_0$. Note that $Q(\cdot,\cdot)$ satisfies the following Bellman equation:
\begin{equation}
    Q(\bx_0,\bu_0) = c(\bx_0,\bu_0) + \gamma J(f(\bx_0,\bu_0))\quad\forall (\bx_0, \bu_0) \in \mathbb{X}\times\mathbb{U}\;, \label{eq:q-bellman-eq}
\end{equation}
and furthermore:
$$
Q(\bx, \bu) \geq J(\bx)\quad\forall \bx\in\mathbb{X}\;,\quad\quad\min_{\bu\in\mathbb{U}}Q(\bx, \bu) = J(\bx)\quad\forall \bx\in\mathbb{X}\;.
$$
In particular, with the definitions above, since $(\bx_i,\bu_i)$ comes as an optimal tuple in the dataset $\mathcal{D}$ we have:
$$
Q(\bx_i,\bu_i)=J(\bx_i)\;.
$$

Then, the following string of (in)equalities hold, as explained below
\begin{align}
    \mathcal{T}^\pi\jub(\bx) - \jub(\bx) &= {\color{red}c(\bx, \bu_i)} + \gamma \jub(\bx') - \jub(\bx) \\
    & = {\color{red}Q(\bx,\bu_i) - \gamma J(\bx')} + \gamma \jub(\bx') - {\color{blue}\jub(\bx)} \\
    & = {\color{green!50!black}Q(\bx,\bu_i)} - \gamma J(\bx') + \gamma \jub(\bx') - {\color{blue}\bj_i - \lambda\|\bx-\bx_i\|} \\
    &\leq {\color{green!50!black}\bj_i + L_J\|\bx-\bx_i\|} + \gamma{\color{purple}\left(\jub(\bx')-J(\bx')\right)}-\bj_i -\lambda \|\bx-\bx_i\| \\
    &\leq (L_J-\lambda)\|\bx-\bx_i\| + \gamma{\color{purple}\left(\bj_j + \lambda\|\bx'-\bx_j\|-\bj_j+ L_J\|\bx'-\bx_j\|\right)}\\
    &= (L_J-\lambda)\|\bx-\bx_i\| + \gamma(\lambda+L_J)\|\bx'-\bx_j\|\\
    &\leq (L_J-\lambda)\|\bx-\bx_i\| + \gamma (\lambda+L_J)L_f\|\bx-\bx_i\|\\
    &\leq 0 \iff
\end{align}
\begin{align}
&\iff (L_J-\lambda) + \gamma (\lambda+L_J)L_f \leq 0 \\
&\iff L_J + \gamma L_J L_f \leq \lambda (1-\gamma L_f) \\
& \iff \lambda \geq \frac{1+\gamma L_f}{1 - \gamma L_f}L_J \;,
\end{align}
\revise{where {\color{red}the red identity} follows from \eqref{eq:q-bellman-eq}, {\color{blue} the blue identity} is the definition of $\jub(\cdot)$, {\color{green!50!black} the green inequality} follows from an upper bound on $Q(\cdot,\cdot)$. Letting $\bx_j=f(\bx_i,\bu_i)$, which belongs to $\mathcal{D}$ under the hypothesis that it contains trajectories, {\color{purple}the second inequality} follows from the definition of $\jub(\bx)$ and a lower bound on $J(\bx')$, using $\bx_j$ in the norms (a valid upper bound, even if index $j$ does not minimize $\jub(\bx')$). Then, the Lipschitz continuity of $f$:
$$
\left\|\bx'-\bx_j\right\| = \left\|f(\bx,\bu_i)-f(\bx_i,\bu_i)\right\|
\leq L_f\|\bx-\bx_i\|
$$ gives the third inequality.}
\end{proof}

\subsection{Theorem \ref{thm:performance-guarantees}} \label{app:proof-performance-guarantees}
\begin{proof}
Note that the expression \eqref{eq:relative-performance} we want to show is equivalent to:
\begin{equation}(1+\beta)\,(J(\bx,\varepsilon)+\eta) \;\ge\; J^\pi(\bx)+\eta. \quad\forall\bx\in\mathbb{X}\;.\label{eq:thm2-prop-subopt}
\end{equation}
Recall as well (Def. \ref{def:J-ub}) that $\jlb(\cdot)$ is a lower bound of $J(\bx,\varepsilon)$:
$$
J(\bx,\varepsilon) \geq \jlb(\bx) = \max_{1\leq i \leq |\mathcal{D}|}\big\{\bj_i+\eta-\lambda\|\bx-\bx_i\|\big\},$$
and that $\jub(\bx)\geq J^\pi(\bx)$ since we are in the conditions of Theorem \ref{thm:pi-eval}. Therefore, a sufficient condition for \eqref{eq:thm2-prop-subopt} is the inequality highlighted in {\color{red}red} below.

\begin{equation}
    (1+\beta)\left(J(\bx,\varepsilon)+\eta\right) \overset{\text{(Def.~\ref{def:J-ub})}}{\geq} \left(1+\beta\right)\left(\jlb(\bx)+\eta\right)~{\color{red}\geq}~\jub(\bx)+\eta \overset{\text{(Thm.~\ref{thm:pi-eval})}}{\geq} J^\pi(\bx)+\eta
\end{equation}
Fix $\bx\in\mathbb{X}$ and let $\bx_i$ be the state in the dataset that maximizes $\jlb(\bx)$. Note that we can further upper bound $\jub(\bx)$ using $\bx_i$:
$$
\jub(\bx) = \min_{1\leq k\leq|\mathcal{D}|} \left\{\bj_k + \lambda\|\bx-\bx_k\| \right\} \leq \bj_i+\lambda\|\bx-\bx_i\|\;.$$
Our sufficient condition then becomes:
\begin{align}
    \left(1+\beta\right)\left(\jlb(\bx)+\eta\right) &= \left(1+\beta\right)\left(\bj_i+\eta-\lambda\|\bx-\bx_i\|\right)~{\color{red}\geq}~\bj_i + \lambda\|\bx-\bx_i\|+\eta \implies
\end{align}
\begin{equation}
  \frac{\beta}{(2+\beta)\lambda}(\bj_i+\eta) \geq \|\bx-\bx_i\|\;, \label{eq:towards-err}
\end{equation}
which is \eqref{eq:relative-performance}, completing the first half of our proof.
To show (\ref{eq:true-relative-performance}), we have the following decomposition:
\begin{align*}\sup_{\bx\in\mathbb{X}}~\frac{J^\pi(\bx) - J(\bx)}{J(\bx)+\eta}&=\sup_{\bx\in\mathbb{X}}\left\{\frac{J^\pi(\bx) - J(\bx,\varepsilon)}{J(\bx,\varepsilon)+\eta}\frac{J(\bx,\varepsilon)+\eta}{J(\bx)+\eta}+\frac{J(\bx,\varepsilon)-J(\bx)}{J(\bx)+\eta}\right\}\\
&=\sup_{\bx\in\mathbb{X}}\left\{\left(\frac{J^\pi(\bx) - J(\bx,\varepsilon)}{J(\bx,\varepsilon)+\eta}+1\right)\frac{J(\bx,\varepsilon)+\eta}{J(\bx)+\eta}-1\right\}\\
&\overset{(i)}{\leq}\sup_{\bx\in\mathbb{X}}\left\{\left(\sup_{\bx\in\mathbb{X}}\left\{\frac{J^\pi(\bx) - J(\bx,\varepsilon)}{J(\bx,\varepsilon)+\eta}\right\}+1\right)\frac{J(\bx,\varepsilon)+\eta}{J(\bx)+\eta}-1\right\}\\
&\overset{(ii)}{\leq}\sup_{\bx\in\mathbb{X}}\left\{\left(\beta+1\right)\frac{J(\bx,\varepsilon)+\eta}{J(\bx)+\eta}-1\right\}\\
&=\left(\beta+1\right)\sup_{\bx\in\mathbb{X}}\left\{\frac{J(\bx,\varepsilon)+\eta}{J(\bx)+\eta}\right\}-1\\
&\overset{(iii)}{\leq}\left(\beta+1\right)\sup_{\bx\in\mathbb{X}}\left\{\frac{J(\bx,\varepsilon)+\eta}{J(\bx,\varepsilon)+\eta-L\varepsilon}\right\}-1\\
&=\left(\beta+1\right)\left(1+\sup_{\bx\in\mathbb{X}}\left\{\frac{L\varepsilon}{J(\bx,\varepsilon)+\eta-L\varepsilon}\right\}\right)-1\\
&=(\beta+1)\left(1+\frac{L\varepsilon}{\eta-L\varepsilon}\right)-1 \\
&=\frac{(\beta+1)\eta-\eta+L\varepsilon}{\eta-L\varepsilon}\\
&=\frac{\beta\eta-L\varepsilon}{\eta-L\varepsilon}
\end{align*}
where $(i)$ follows from the nonnegativity 
$\frac{J^\pi(\bx)-J(\bx,\varepsilon)}{J(\bx,\varepsilon)+\eta}\ge 0$ and 
$\frac{J(\bx,\varepsilon)+\eta}{J(\bx)+\eta}\geq 1$. 
Using \eqref{eq:relative-performance} yields $(ii)$. 
For $(iii)$, if $\varepsilon$ is small enough such that $J(\bx,\varepsilon)+\eta-L\varepsilon>0~\forall\bx\in\mathbb{X}$, we invoke the first item of Assumption~\ref{assn:J-locally-lipschitz} 
to lower bound the denominator in $\frac{J(\bx,\varepsilon)+\eta}{J(\bx)+\eta}$. 
Rearranging terms, we finish the proof.
\end{proof}
\clearpage
\subsection{Proposition \ref{prop:sample-complexity}}\label{proof:4}
\begin{proof}
We split the proof in two parts.
\paragraph{1. A $2r-$cover of $\mathbb{X}$ yields a recursively feasible and $\beta-$optimal policy.}
The radius $r$ in the proposition satisfies:
$$2r = \min\left\{\frac{\eta}{\lambda}\cdot\frac{\beta\eta-L\varepsilon(1+\beta)}{(2+\beta)\eta-L\varepsilon(1+\beta)},\,\frac{\varepsilon}{L_f}\right\}.$$
First, note that if $\bigcup_{i=1}^{|\mathcal{D}|}\mathbb{B}(\bx_i,2r) \supseteq \mathbb{X}$, then policy $\pi_{\mathcal{D}}$ is recursively feasible: this is because $\forall \bx, \exists \bx_i : \|\bx-\bx_i\|\leq 2r \leq \tfrac{\varepsilon}{L_f}$ and we invoke the result of Proposition \ref{prop:recursive-feasibility}.

We now show that a $2r-$cover also implies that $\pi_{\mathcal{D}}$ is $\beta-$optimal. Recall the result of Theorem \ref{thm:performance-guarantees}: if the dataset $\mathcal{D}$ has sufficient coverage, in the sense that for any $\bx\in\mathbb{X}$ there exists $\bx_i\in\mathcal{D}$ such that:
\begin{align}
\|\bx-\bx_i\|\leq\frac{\beta'}{(2+\beta')\lambda}(\bj_i+\eta)&\implies \label{eq:pf-thm-2}\\    
&\implies\sup_{\bx\in\mathbb{X}}\frac{J^\pi(\bx)-J(\bx,\varepsilon)}{J(\bx,\varepsilon)+\eta}\leq\beta'\\
&\implies \sup_{\bx\in\mathbb{X}}~\frac{J^\pi(\bx) - J(\bx)}{J(\bx)+\eta} \leq\frac{\beta'\eta+L\varepsilon}{\eta-L\varepsilon}
\end{align} 
where $\beta' > 0$ and $\eta > L\varepsilon$. We want to use the smallest $\beta'$ such that the right-most term is upper-bounded by $\beta$, i.e.:

\begin{align}
\frac{\beta'\eta + L\varepsilon}{\eta - L\varepsilon}&\leq \beta \implies
\beta'\leq\frac{\beta(\eta - L\varepsilon)-L\varepsilon}{\eta}\label{eq:150}\;.
\end{align}
We plug this value of $\beta'$ back in \eqref{eq:pf-thm-2} and upper bound the right hand side using the fact that $\bj_i\geq 0$:
$$
\|\bx-\bx_i\|\leq\frac{\beta'}{(2+\beta')\lambda}
$$
that is to say:
\begin{align}
\|\bx-\bx_i\|&\leq \frac{\frac{\beta(\eta - L\varepsilon)-L\varepsilon}{\eta}}{\lambda\left(2+\frac{\beta(\eta - L\varepsilon)-L\varepsilon}{\eta}\right)}\eta\nonumber\\&= \frac{(\beta(\eta - L\varepsilon)-L\varepsilon)\eta}{\lambda\left(2\eta+\beta(\eta - L\varepsilon)-L\varepsilon\right)},\label{eq:160}\\
&= \frac{\eta}{\lambda}\cdot\frac{\beta\eta-L\varepsilon(1+\beta)}{(2+\beta)\eta-L\varepsilon(1+\beta)}\;.
\end{align}
This finishes the first part of the proof.
\paragraph{2. Sample complexity of Algorithm \ref{alg:1} to achieve a 2r-cover of $\mathbb{X}$}
In the prequel we showed that a $2r-$cover of $\mathbb{X}$ yields a recursively feasible and $\beta-$optimal policy. Now the question is: if we sample $\{\bx_i\}_{i=1}^n$ uniformly from $\mathbb{X}$, what is the sample complexity of $n$ to get a $2r-$cover?

Consider an arbitrary $r-$cover of $\mathbb{X}$, i.e. center points 
$\{\by_1,\ldots,\by_m\}$ such that:
$$
\bigcup_{k=1}^m \mathbb{B}(\by_k, r) \supseteq \mathbb{X}\;.
$$
Such a cover exists because $\mathbb{X} $ is compact, and $m\approx \operatorname{N}_{\text{cover}}\left(\mathbb{X}; r\right)$. Let's call $B_k \triangleq \mathbb{B}(\by_k,r), k=1,\ldots m$. Observe that, if the following condition holds:
\begin{equation}
\forall B_k~\exists \bx_i \in \mathcal{D} : \bx_i \in B_k  \label{eq:1-sample-per-ball}
\end{equation}
we can conclude
$$
\|\bx-\bx_i \|\leq 2r~\forall \bx\in\mathbb{X},
$$
that is to say: if there is at least one sample $\bx_i$ \emph{in every} ball $B_k$, then $\{\bx_i\}_{i=1}^n$ form a $2r-$cover of $\mathbb{X}$. Indeed, by the triangle inequality:
\begin{align}
\|\bx-\bx_i\|\leq\|\bx-\by_k\|+\|\bx_i-\by_k\|\leq 2r
\end{align}
Therefore, the sample complexity of Algorithm \ref{alg:1} is upper-bounded by the sample complexity of guaranteeing \eqref{eq:1-sample-per-ball}.  Given samples $\bx_i\overset{iid}{\sim} \operatorname{Unif}(\mathbb{X})$, observe that: 
$$\mathbb{P}(\bx_i\in B_k)\approx \frac{1}{m}=\frac{1}{\operatorname{N}_{\text{cover}}\left(\mathbb{X}; r\right)}\;.$$
Now, fix $k\in\{1,\ldots, m\}$ and assume Algorithm \ref{alg:1} has been run for $n$ rounds.
\begin{align}
    \mathbb{P}\left(\nexists i\in\{1,\ldots,n\}: \bx_i  \in B_k \right) \lesssim \left(\frac{1}{\operatorname{N}_{\text{cover}}\left(\mathbb{X}; r\right)}\right)^n \leq \exp\left(-\frac{n}{\operatorname{N}_{\text{cover}}\left(\mathbb{X}; r\right)}\right)
\end{align}
We then apply a union bound over $k$ and upper bound that probability by $\delta\in(0, 1)$:
\begin{align}
    \mathbb{P}\left(\text{some ball~}B_k\text{~not covered by any~}\bx_i\right) &= \mathbb{P}\left(\bigcup_{k=1}^m\left\{\nexists i\in\{1,\ldots,n\}: \bx_i  \in B_k\right\}\right) \\
    &\leq \sum_{i=1}^m \mathbb{P}\left(\nexists i\in\{1,\ldots,n\}: \bx_i  \in B_k \right) \\
    &\lesssim {\operatorname{N}_{\text{cover}}\left(\mathbb{X}; r\right)}\exp\left(-\frac{n}{\operatorname{N}_{\text{cover}}\left(\mathbb{X}; r\right)}\right)\\
    &\leq \delta \implies
\end{align}
\begin{align}
    n \geq \operatorname{N}_{\text{cover}}\left(\mathbb{X}; r\right) \log\left(\operatorname{N}_{\text{cover}}\left(\mathbb{X}; r\right)\right) \log\left(\frac{1}{\delta}\right)\;,
\end{align}
as desired.

\end{proof}
\section{Considerations on Algorithm \ref{alg:2}}\label{appen:c5}

\begin{algorithm}[ht]
\caption{Local Adaptive Data collector} \label{alg:2}
\KwIn{$\varepsilon>0$, $T>0$. Initial covering radius $h_0$. Desired gap $\beta$. Slack $\eta$.}

\KwSty{Initialize:} $\mathcal{D}=\emptyset$. $\mathrm{Tree}(0)=\mathbb{X}$.
$I=\mathbf{0}_{1\times\left\lceil\log_{3}\left(\frac{h_0 L_f}{\varepsilon}\right)\right\rceil}.$

Sample $\bx$ through the uniform grid $\{\mathbb{X}^{1,1},\mathbb{X}^{1,1},\mathbb{X}^{1,2},...,\mathbb{X}^{1,N_0}\}=\mathcal{H}_0$ with radius $h_0$.

Get trajectory ${\tau}_{1,i}=\{\left(\bx_{1,m},\bu_{1,m}, \bj_{1,m}\right)\}_{m=i}^{T+i-1}$ by solving \eqref{eq:conservative-prob} from $\bx_{1,i}\in\mathbb{X}^{1,i}\subseteq\mathcal{H}_0$.
    
$\mathcal{D} \gets \mathcal{D} \cup \{\left(\bx_{1,m},\bu_{1,m}, \bj_{1,m}\right)\}$~\textbf{for}~$(\bx_{1,m},\bu_{1,m}, \bj_{1,m})$ \textbf{in}~${\tau}_{1,i}$. $I_1\gets N_0$. $\mathrm{Tree}(1)\gets\mathcal{H}_0$. $t\gets1$.

   \While{$\exists \,\mathbb{X}^{k,j}\subseteq\mathrm{Tree}(t)$ infeasible,}{
   
   \For{$\mathbb{X}^{k,j}\subseteq\mathrm{Tree}(t)$}{
   
   Check the one-step feasibility (Equation (\ref{eq:feasibility-radius}), Theorem \ref{prop:local-feasibility}).
   
   \If{$\mathbb{X}^{k,j}$ one-step infeasible,}{

\SetKwBlock{BlkOne}{\textbf{Block 1:}}{\textbf{End Block 1}}

    \BlkOne{Split the cell into the uniform grid $\{\mathbb{X}^{k+1,I_{k+1}+1},...,\mathbb{X}^{k+1,I_{k+1}+3^n}\}=\mathcal{H}_{k,j}$.

   $\mathrm{Tree}(t)\gets\mathrm{Tree}(t)\cup\mathcal{H}_{k,j}/\mathbb{X}^{k,j}$. $I_{k+1}=I_{k+1}+3^n$.
   
   \For{$\mathbb{X}^{k+1,i}\subseteq\mathcal{H}_{k,j}$}{
   Sample $\bx_{k+1,i}$ as the central point of $\mathbb{X}^{k+1,i}$.

Get trajectory ${\tau}_{k+1,i}=\{\left(\bx_{k+1,m},\bu_{k+1,m}, \bj_{k+1,m}\right)\}_{m=i}^{T+i-1}$ by solving \eqref{eq:conservative-prob} from $\bx_{k+1,i}\in\mathbb{X}^{k+1,i}\subseteq\mathcal{H}_{k,j}$.

    $\mathcal{D} \gets \mathcal{D} \cup \{\left(\bx_{k+1,m},\bu_{k+1,m}, \bj_{k+1,m}\right)\}$~\textbf{for}~$(\bx_{k+1,m},\bu_{k+1,m}, \bj_{k+1,m})$ \textbf{in}~${\tau}_{k+1,i}$.}

  }}}

   $\mathrm{Tree}(t+1)\gets\mathrm{Tree}(t)$. $t\gets t+1$.
   
   }
   
   \While{$N\geq K$,}{
   
   Compute the number of cells $\mathbb{X}^{k,j}\subseteq\mathrm{Tree}(t)$ not $\beta$-optimal as $N_{\mathrm{Tree}(t)}$. 

   \eIf{$3^n N_{\mathrm{Tree}(t)}\leq N-K$,}
   {\For{$\mathbb{X}^{k,j}\subseteq\mathrm{Tree}(t)$}{

   Check the $\beta$-optimality (Last equation in Theorem \ref{thm:performance-guarantees}).
   
   \If{$\mathbb{X}^{k,j}$ not $\beta$-optimal policy,}{Run \textbf{Block 1}. $K\gets K+3^n$.}
   
   \eIf{$\exists \,\mathbb{X}^{k,j}\subseteq\mathrm{Tree}(t)$ not $\beta$-optimal,}{$\mathrm{Tree}(t+1)\gets\mathrm{Tree}(t)$. $t\gets t+1$.}{\textbf{End All}}
   }}{
   
   Run \textbf{Block 1} for $\lfloor \frac{N-K}{3^n}\rfloor$ cells
   
   $K\gets K+\left(\lfloor \frac{N-K}{3^n}\rfloor+1\right)3^n$.}
   }
\end{algorithm}
\clearpage
\begin{assumption}[Regularity of $\mathbb{X}$ and $\|\cdot\|_{\mathbb{X}}$]
$\mathbb{X}$ is a hypercube containing the origin and $\|\cdot\|_{\mathbb{X}}=\|\cdot\|_{\infty}$.\label{assu:15}
\end{assumption}
Algorithm \ref{alg:2} mainly follows the idea of Neyman allocation \citep{garivier2016optimal,adusumilli2022neyman,fogliato2024framework} and deterministic non-parametric optimization \citep{munos2011optimistic,kuzborskij2020locally,sinclair2020adaptive}. We construct a refinement tree whose nodes are axis–aligned hypercubes ($\ell_{\infty}$-balls) 
\(\mathbb{X}^{k,j}\subseteq\mathbb{X}\), 
where \(k\) denotes the tree layer (depth) and \(j\) indexes the nodes at the layer $k$ of the current $\mathrm{Tree}(t)$.
A layer-\(k\) hypercube has radius \(h_k\) (side length \(2h_k\)). 
Splitting replaces a node by \(3^n\) children, each of radius \(h_{k+1}=h_k/3\). This splitting rule follows the logic of \citet[Algorithm 2]{siegelmann2025stability}, we split each edge into $\frac{1}{3}$ because the central point still remains a central point in the new subgrid $\mathcal{H}_{k,j}$\footnote{For example, if we split each edge of the cell $\mathbb{X}^{k,j}$ by half, then the central point, $\bx_{k,j}$, is not usable for any children cell $\mathbb{X}^{k+1,i}$ in the new subgrid $\mathcal{H}_{k,j}$.}. Specifically, a list $I=\left\{I_1,I_2,...,I_{\left\lceil\log_{3}\left(\frac{h_0 L_f}{\varepsilon}\right)\right\rceil}\right\}$ is defined initial to keep on track of the number of leaves in different layer $k$ in current $\mathrm{Tree}(t)$. The length of list is $\left\lceil\log_{3}\left(\frac{h_0 L_f} {\varepsilon}\right)\right\rceil$ because $h_0$ can split at most this time to reach a sufficient small radius $\frac{\varepsilon}{L_f}$ in Proposition \ref{prop:recursive-feasibility}. The Algorithm~\ref{alg:2} proceeds in two stages:

(a) we first construct a nonparametric policy that guarantees one-step feasibility;

(b) given a budget \(N\), we adaptively split cells to achieve \(\beta\)-optimality relative to the best achievable performance.

\emph{Initialization.} 
We initialize with a uniform cover $\mathcal{H}_0=\{\mathbb{X}^{1,i}\}_{i=1}^{N_0}$ of common radius $h_0$ and cardinality $N_0$. 
For each cell $\mathbb{X}^{k,j}$ with center $\bx_{k,j}$, we evaluate the dynamics along the horizon $m=i,\dots,i+T-1$ by solving~\eqref{eq:conservative-prob}, obtaining the control inputs $\bu_{k,m}$ and successor states $\bj_{k,m}$. 
The resulting optimal trajectory is
$\tau_{k,i}=\big\{(\bx_{k,m},\,\bu_{k,m},\,\bj_{k,m})\big\}_{m=i}^{T+i-1}$.

\emph{Acceptance/Refinement rule.}
A cell $\mathbb{X}^{k,j}$ is \emph{accepted} if
(i) it is one-step feasible for stage (a), or
(ii) it is $\beta$-optimal policy for stage (b) at its center-trajectory
$\tau_{k,i}=\{(\bx_{k,m},\bu_{k,m},\bj_{k,m})\}_{m=i}^{i+T-1}$, respectively.
If both (i) and (ii) hold simultaneously, the cell $\mathbb{X}^{k,j}$ is
\emph{permanently kept}.
Otherwise, $\mathbb{X}^{k,j}$ is \emph{refined}:
we split it into with subgrid $\mathcal{H}_{k,j}$ with $3^n$ uniform children of radius $h_{k+1}=h_k/3$ and replace
the parent in the current tree $\mathrm{Tree}(t)$ by these children (\textbf{Block~1}). The central trajectory $\tau_{k+1,i}$ for each child cell $\mathbb{X}_{k+1,i}$ in $\mathcal{H}_{k,j}$ is computed and collected in the data set $\mathcal{D}$.

In stage (b), whenever a cell is refined, we update the evaluation
budget by $K \gets K + 3^n$.

\emph{Iteration and stopping.}
At iteration $t$, we traverse all current leaves of $\mathrm{Tree}(t)$ and apply the acceptance/refinement rule to each leaf; the current iteration $t$ ends once all leaves have been processed.
The procedure terminates in stage (a) as soon as the resulting tree is one-step feasible.
For stage (b), the algorithm terminates when either the evaluation budget is exhausted, i.e., $N-K<3^n$, or every leaf in current $\mathrm{Tree}(t)$ is already $\beta$-optimal policy.

In stage (b), if the remaining budget $N-K$ is insufficient to refine all non–$\beta$-optimal leaves in $\mathrm{Tree}(t)$, we refine only 
$\left\lfloor \frac{N-K}{3^n} \right\rfloor$ leaves—specifically, those with the largest relative error bounds.
\begin{remark}[Cover a regular domain (Assumption \ref{assu:15}).] \label{rem:4}In Assumption \ref{assu:15}, we assume $\mathbb{X}$ is a regular hypercube to ensure a uniform grid to cover $\mathbb{X}$ without overlapping or omission. This assumption here is just to ease the clarfication of our presentation on Algorithm \ref{alg:2}. The Algorithm \ref{alg:2} could be extend to any shape of $\mathbb{X}$ easily by consider a feasible outer cover $\mathbb{X}_0$ ($\forall \bx\in\mathbb{X}_0,\,J(\bx,\varepsilon)<+\infty$) constructed by non-overlapping same radius $\ell_{\infty}$-balls.
\end{remark}

\section{Details on the experiments}\label{app:experiments}
\subsection{Inverted Pendulum}
We consider the inverted pendulum with angle $x_1$ and angular velocity $x_2$, with equations of motion given by:
\begin{equation}
\dot\bx = \frac{d}{dt}\begin{bmatrix}x_1\\x_2\end{bmatrix} = \begin{bmatrix}x_2\\-\frac{g}{l}\sin x_1\end{bmatrix} + \frac{1}{ml^2}\begin{bmatrix}0\\1\end{bmatrix}u \;,\label{eq:pendulum-dynamics}
\end{equation}
where $u\in\mathbb{R}$ is the scalar torque applied to the axis of the pendulum. We consider the discrete-time dynamics of \eqref{eq:pendulum-dynamics} with sampling time $\delta t=0.05s$, we let $\bu_t\in\mathbb{U}=[-5, 5], m=1kg, l=1m, g=9.82 \tfrac{m}{s^2}$. The stage cost is given by:
$$
c(\bx_t,\bu_t) = \bx_t^\top \begin{bmatrix}1&0\\0&0.1\end{bmatrix}\bx_t + 0.01\|\bu_t\|^2.$$
We consider a time horizon $T=100$, and obtain different datasets $\mathcal{D}$ by performing a uniform grid of the state space $\mathbb{X} = [-2, 2]^2$ with $G\in\{5, 7, 9,11\}$ points per dimension, and run offline MPC to get full length trajectories starting from those points. 

\subsection{Minimum time with control regularization}
\begin{equation}
\dot\bx = \frac{d}{dt}\begin{bmatrix}x_1\\x_2\end{bmatrix} = \begin{bmatrix}0.2 & 1\\ 0 & 0\end{bmatrix}\begin{bmatrix}x_1\\x_2\end{bmatrix} + \begin{bmatrix}0\\1\end{bmatrix}u \label{eq:min-time-dynamics}
\end{equation}
We consider the discrete-time version of \eqref{eq:min-time-dynamics} with sampling time $\delta t=0.1s$, stage cost:
\begin{equation}
    c(\bx_t,\bu_t) = 1+ 10\cdot\|\bu_t\|\;.
\end{equation}
and terminal constraint $\bx_T = 0$. Even though this cost functions are not captured by our theory above, we want to show that our policy still achieves good performance in this case. We consider the discrete-time dynamics of \eqref{eq:min-time-dynamics} with sampling time $\delta t=0.1s$, horizon $T=200$ and do uniform gridding of the state space $\mathbb{X} = [-2, 2]^2$ with the same grids as for the pendulum, $\mathbb{U}=[-1, 1]$.

\subsection{Verification on constrained LQR}
\begin{equation}
    A = \begin{bmatrix}
        1 & 0.1 \\ 
        0 & 1
    \end{bmatrix};\quad\quad B = \begin{bmatrix}0.15\\1\end{bmatrix},
\end{equation}
We consider stage cost $c(\bx,\bu) = \|\bx\|^2 + \|\bu\|^2$, constraint sets $\mathbb{X}=[-3, 3]^2, \mathbb{U} = [-2, 2]$ and a horizon $T=10$. For our policy we used $L_f=1.1, \lambda=1, \beta=5, \eta=0.01$. This parameters were chosen in an ad-hoc manner to get a convergent result in $5$ iterations or less---both for feasibility and performance---so that Figure \ref{fig:feasibility-experiments} was illustrative. With tighter requirements (e.g. $\beta=0.1, \eta=1e-6$) Algorithm \ref{alg:verification} necessitates many more iterations.
\end{document}